\documentclass[11pt,final]{article}
\usepackage[letterpaper, margin=1in]{geometry}

\usepackage[utf8]{inputenc}
\usepackage[final]{microtype}
\usepackage{amsmath,amsthm,thmtools,mathtools,breqn,amssymb,amsbsy,algorithm,thm-restate,tikz}
\usepackage{complexity}
\usetikzlibrary{automata, positioning, arrows}
\usepackage[noend]{algpseudocode}
\usepackage[final,hidelinks]{hyperref}
\usepackage[capitalize,nameinlink]{cleveref}
\usepackage[inline]{enumitem}
\makeatletter
\let\c@author\relax
\makeatother
\usepackage[url=false,isbn=false,style=alphabetic,citetracker=true,sorting=nyt,mincitenames=10,maxcitenames=10,maxbibnames=99,backend=bibtex]{biblatex}
\allowdisplaybreaks
\DeclareDatafieldSet{setall}{
  \member[datatype=literal]
  \member[datatype=name]
  \member[field=journal]
}
\DeclareSourcemap{
  \maps[datatype=bibtex]{
    \map[overwrite, foreach={setall}]{
      \step[fieldsource=\regexp{$MAPLOOP},
            match=\regexp{\x{0131}\x{0300}},
            replace=\regexp{\x{00EC}}]
      \step[fieldsource=\regexp{$MAPLOOP},
            match=\regexp{\x{0131}\x{0301}},
            replace=\regexp{\x{00ED}}]
      \step[fieldsource=\regexp{$MAPLOOP},
            match=\regexp{\x{0131}\x{0302}},
            replace=\regexp{\x{00EE}}]
      \step[fieldsource=\regexp{$MAPLOOP},
            match=\regexp{\x{0131}\x{0308}},
            replace=\regexp{\x{00EF}}]
    }
  }
}

\numberwithin{equation}{section}
\declaretheorem[Refname={Theorem,Theorems}]{theorem}

\declaretheorem[sibling=theorem,Refname={Lemma,Lemmas}]{lemma}
\declaretheorem[sibling=theorem,Refname={Claim}]{claim}
\declaretheorem[sibling=theorem,Refname={Proposition,Propositions}]{proposition}
\declaretheorem[sibling=theorem,Refname={Definition,Definitions}]{definition}

\let\deg\relax 
\DeclareMathOperator{\adj}{\mathsf{adj}}
\DeclareMathOperator{\deg}{\mathsf{deg}}
\DeclareMathOperator{\query}{\mathsf{query}}
\newcommand{\ceil}[1]{\left \lceil #1 \right \rceil}
\newcommand{\floor}[1]{\left \lfloor #1 \right \rfloor}

\renewcommand{\E}[1]{\mathbb{E}{#1}}

\newcommand{\Probe}{\mathrm{Probe}}
\newcommand{\Foot}{\mathrm{Foot}}
\newcommand{\Bin}{\mathrm{Bin}}

\renewcommand{\epsilon}{\varepsilon}
\newcommand{\patrascu}{P\v{a}tra\c{s}cu}
\DeclareMathOperator{\ans}{\mathsf{ans}}

\title{Cell-Probe Lower Bound for Accessible Interval Graphs}

    \author{Sankardeep Chakraborty\thanks{\texttt{sankardeep.chakraborty@gmail.com.}}
    \\{University of Tokyo,  Japan}
    \and{Christian Engels\thanks{\texttt{christian@nii.ac.jp}. This work was supported by JSPS KAKENHI Grant Number JP18H05291.}}
    \\{National Institute of Informatics,  Japan}
    \and{Seungbum Jo\thanks{\texttt{sbjo@cnu.ac.kr}.}}
    \\{Chungnam National University, South Korea}
    \and{Mingmou Liu\thanks{\texttt{mili@di.ku.dk}. Mingmou is a part of Basic Algorithms Research Copenhagen (BARC), supported by the VILLUM
    Foundation grant 16582. Part of
    this work was done when the author
        was a research fellow at the Nanyang Technological University, supported by the Singapore Ministry of Education (AcRF) Tier 1 grant
        RG75/21.}}
    \\{University of Copenhagen, Denmark}}

\begin{document}

\maketitle

\thispagestyle{empty}

\begin{abstract}
    We spot a hole in the area of succinct data structures for graph classes from a universe of size at most $n^n$. Very often, the input
    graph is labeled by the user in an arbitrary and easy-to-use way, and the data structure for the graph relabels the input graph in some
    way. For any access, the user needs to store these labels or compute the new labels in an online manner. This might require more bits
    than the information-theoretic minimum of the original graph class, hence, defeating the purpose of succinctness. Given this, the data
    structure designer must allow the user to access the data structure with the original labels, i.e., relabeling is not allowed. We call
    such a graph data structure ``accessible''. In this paper, we study the complexity of such accessible data structures for interval
    graphs, a graph class with information-theoretic minimum less than $n\log_2 n$ bits. More specifically, our contributions are threefold:
    \begin{itemize}
        \item We formalize the concept of ``accessibility'' (which was implicitly assumed in all previous succinct graph data structures),
              and propose a new presentation, called ``universal interval representation'', for interval graphs.
        \item Any data structure for interval graphs in universal interval representation, which supports both adjacency and degree query
              simultaneously with time cost $t_1$ and $t_2$ respectively, must consume at least $\log_2(n!)+n/(\log n)^{O(t_1+t_2)}$ bits of space.
              This is also the first lower bound for graph classes with information-theoretic minimum less than $n\log_2n$ bits.
        \item Finally, we provide two efficient succinct data structures for interval graphs in universal interval representation that
              support adjacency query and degree query individually in constant time and space costs $\log_2(n!) + \tilde{O}(\sqrt n)$ bits and
              $\log_2(n!) + \tilde{O}(n^{2/3})$ bits respectively.
    \end{itemize}

    Therefore, two upper bounds together with the lower bound show that the two elementary queries for interval graphs are incompatible with
    each other in the context of succinct data structure. To the best of our knowledge, this is the first proof of such incompatibility
    phenomenon.
\end{abstract}

\setcounter{page}{1}

\section{Introduction}\label{sec:intro}
In the field of \textit{compressed data structures} design, given a combinatorial structure $T$ as input, e.g., a graph, a poset, a
permutation, etc., a storage scheme is
called \textit{optimal} if it takes an information-theoretic minimum
of $\log(|T|)$ bits (throughout this paper, $\log$ denotes the logarithm to the base $2$) to encode any arbitrary member $x \in T$. In
general, it is not difficult to
encode any $x \in T$ optimally, but, it is extremely difficult to navigate and query on $x$ efficiently in this compressed form. The whole
point of compressed data
structures is to represent combinatorial objects with as few as possible bits, not just for the sole purpose of compressing them, but also
to
query them efficiently without decompressing.
To this end, it is customary to allow some
extra space (called ``redundancy'' in compressed data
structure parlance) so
that queries can be executed
on the code word of $x$ efficiently.
This gives rise to the notion of \textit{succinct} data structures. Formally, a data structure is called \textit{succinct} if the total space
used is
$\log(|T|)+ o(\log(|T|))$ bits. Since Jacobson~\cite{Jacobson89} initiated the investigation of succinct data structures for various classes
of graphs 35 years ago,
extensive research in this domain has led to such data structures (refer to~\cite{gonzalo} for a thorough introduction) for various
combinatorial objects like
trees~\cite{NavarroS14}, arbitrary graphs~\cite{FarzanM13}, planar maps~\cite{AleardiDS08}, finite
automaton~\cite{ChakrabortyGSS23}, functions~\cite{MunroR04}, permutations~\cite{MunroRRR03}, posets~\cite{MunroN16}, bounded treewidth
graphs~\cite{FarzanK11},
texts~\cite{GagieNP20}, sequences~\cite{NavarroN14}, etc., among
many others and by now this is a mature field of study. Typically, the ideal query time would be constant (or polylogarithmic in the input
size) for such data structures,
thus, these data structures
truly offer the best of both worlds (time- and space-wise) solution for
modern storage issues.

We observe that \textit{underlying all compressed data structures}, there exists a fundamentally subtle (yet implicit in previous works)
assumption which we call the
notion of \textit{``accessibility''}. It can be defined from the end user's point of view in the following sense. Suppose that the user has
an unlabeled graph $G$ as
input, so huge
that the user first wants to construct an encoding $E$ of $G$ using little space, and then, the user wants to execute queries
efficiently on $E$ to retrieve useful information about $G$. Note that, to query $G$, the user has to assign some labels to either the
vertices or the edges of $G$ in an
arbitrary and easy-to-use way, and the user always wants to query on the graph with the labels chosen by herself.
From now on, we assume that only
vertices are being labeled/relabeled, not the edges, but in general, one can relabel either or both.
The user then hands over $G$ to the data structure designer, who can either
adopt the user's labeling or relabel the vertices for any arbitrary reasons. When the designer returns the final data structure to the user,
two scenarios emerge.
If the designer used the same labeling as the user, no additional \textit{``vertex mapping translation''} data structure is required,
resulting in an
\textit{``accessible''} data structure from the user's perspective.
However, if the labeling differs, the designer must provide either a vertex mapping translation data structure to the user so that the user
can correctly identify any
vertex in $G$ and can access/query the data structure meaningfully, or an algorithm for the user to compute the labeling by herself in an
online manner.
Note that both the translation data structure and the algorithm could be considered as a built-in component of our data structure for graphs.
Failure to address these cases leads to an \textit{``inaccessible''} data structure for $G$ from the user's perspective.

To illustrate, take as input a labeled tree.
The designer may relabel the tree so that the relabeled tree becomes some unique representative labeled tree in its automorphism group.
This move is very tempting, as it reduces the minimum bits to encode the tree from $(n-2)\log n$ to
$\approx 1.58n$~{\cite[Eq.(58), p. 481]{flajolet2009analytic}}.
However, it makes the data structure inaccessible since it forces the user to spend lots of computational resources (e.g., additional space
as an extra translation data
structure, or lots of time on accessing the original tree to find the correct way to access the data structure) on computing the automorphism
before accessing the data
structure as it is relabeled. Thus, from a data structure designer's perspective, it is certainly not desirable to construct a data structure
that is correct, and
efficient yet inaccessible (thus, somewhat useless) to the end user.
In light of the above discussion, let us now formally define the notion of accessible data structures.
\begin{quote}
    \textit{A data structure is said to be ``accessible'' to a user if the data structure designer uses the labeling chosen by the user for
        the input. Otherwise, the data
        structure is inaccessible.}
\end{quote}
Note that the
above definition naturally carries over to the previously defined succinct data structures, hence, resulting in
\textit{accessible succinct} data structures. While the notion of accessibility is generally a recommended design practice, it is worth
asking: Is it always strictly
required? It turns out that at least in some scenarios, we can relax this notion to some extent. The cost of transforming an ``inaccessible''
structure into an
``accessible'' one, by including a vertex mapping translation data structure, is at most $O(n\log n)$ bits (assuming the input is a graph $G$
with $n$ nodes). This cost is
negligible for combinatorial structures with an information-theoretic lower bound of $\omega(n\log n)$ bits. In such cases, the designer can
easily relabel and provide the
permutation to the user at minimal extra cost. Examples include chordal and split graphs~\cite{MunroW18}, posets~\cite{MunroN16}, strongly
chordal and chordal
bipartite graphs~\cite{Spinrad95}.

The notion of accessibility also demonstrates a simple yet absolutely critical fact that
\textit{the benchmark of a succinct data structure for one graph class is not unique}
as the user may label the graphs in any possible way and the labels should be considered as a part of the input.
More concretely, in the tree example, the information-theoretic minimum of the input varies from $\approx 1.58n$ bits to $(n-2)\log n$ bits,
depending on how the user
labels the tree. Thus, through the lens of ``accessibility'', all the space efficiency of previous data structures should be re-examined.

Concretely, for interval graphs,~\cite{DBLP:journals/algorithmica/AcanCJS21} (we will provide the formal definition and other required
properties of interval graphs shortly)
processes the input
to $n$ intervals in the range of $[1,2n]$ such that there
is exactly one endpoint that lands in each position $i\in[1,2n]$.
The space cost of their data structure is $n\log n+2n$ bits, and the authors show that the information-theoretic minimum of an unlabeled
interval graph is at least $n\log n-2n\log\log n$.
Hence, the authors claim a redundancy of $O(n\log\log n)$ bits.
However, suppose the user labels the graph in the same way as the authors, then its information-theoretic minimum is
$\log\left(\prod_{i=1}^n (2i-1)\right)\approx n\log n-0.44 n$. Thus, its
redundancy is in fact $\approx 2.44n$ bits. \
\textit{In conclusion, we want to emphasize that sticking to the notion of ``accessibility'' is a good data structure design
    principle almost always, and it is the responsibility of the data structure designer to design different data structures adapting to the
    labeling chosen by the
    user. Furthermore, as can be expected, all the existing previous succinct data structures are accessible in some sense, although this
    fact is always implicit in the literature.}

It is through this lens that we study the class of interval graphs in this paper. In particular, in this work,
\textit{we study the complexity of compressed data structures for interval graphs}. Roughly speaking, a simple undirected graph $G$ is called
an \textit{interval graph} if its vertices can be assigned to intervals on
the real line,
so that two vertices are
adjacent in $G$ if and only if their assigned intervals intersect.
Interval graphs are of immensely importance in multiple fields.
\Textcite{Hajos}, and then \textcite{Benser} independently defined these graphs, and since then, interval graphs
have been studied extensively for their structural properties
(for example, interval graphs are chordal graphs and perfect graphs) and algorithmic applications. Interval graphs have appeared in a variety
of applications, for example,
operations research and scheduling theory~\cite{Bar-NoyBFNS01}, bioinformatics~\cite{bio}, temporal reasoning~\cite{GolumbicS93}, VLSI
design~\cite{spin}, and many more.
We refer the reader
to~\cite{Golumbic,Golumbic85} for a thorough introduction to interval graphs and their myriad of applications. Being such a fundamentally
important and practically motivated graph
class, in this
work, we ask: \textit{Can we design an efficient data structure for interval graphs that respects the labeling chosen by the user? What would be the optimal time/space trade-off for such an interval graph data structure?}

\subsection{Our Set-up}\label{sec:setup}

In order to state our main results, let us first set the stage by formally defining and discussing some fundamental concepts and properties
that will be used throughout
this paper. For
any two integers $n$ and $m$, we use $[n]$ to denote $\{1,\dots,n\}$, and use $[m, n]$ to denote $\{m,\dots, n\}$ where $m\le n$.
When we consider data structures for graphs, we are concerned with mainly supporting the adjacency query
and  degree query. Given a graph $G = (V, E)$, and for any two vertices $u, v \in V$, $\adj (u, v)$ query returns true if $u$ and $v$ share
an edge in $G$, and $\deg(v)$
query for a
vertex $v$ returns the number of neighbors of $v$ in $G$. A graph $G = (V, E)$ is defined as an \emph{interval graph} if and only if
\begin{enumerate*}[label=(\roman*)]
    \item for each $v \in V$, there exists a  corresponding interval $I_v$ on the real line, and
    \item for any two vertices $u, v \in V$, $u$ and $v$ are adjacent if and only if $I_u \cap I_v \neq \emptyset$. We call a set of
          intervals $\{I_v \mid v \in G\}$ an  \emph{interval representation} of $G$.
\end{enumerate*}
Note that, a single interval graph can have multiple interval representations. For instance, the complete graph on three vertices,
$K_3$, can be represented as $\{[1,2],[2,3],[1,3]\}$ and $\{[1,1],[1,1],[1,1]\}$. To uniquely define and \textit{access} the vertices of $G$
from its interval representation, we
introduce the notion of a
\emph{universal interval representation} as follows.
\begin{definition}\label{def:universal}
    An interval representation $I$ of an interval graph $G$ with $n$ vertices is a \emph{universal interval representation} if
    every $i\in [n]$ is a leftmost endpoint of exactly one interval.
\end{definition}

Now, from the user's point of view, using the universal interval representation, each vertex can, very naturally, be denoted as the leftmost
endpoint of its corresponding
interval. This simply implies that vertex $i$ corresponds to the interval $I_i = [i, e_i]$. This shows that the universal interval
representation has the potential to be considered as a viable solution for interval graph representation from the \textit{accessibility}
standpoint if the user agrees to the same methodology for referring to the vertices of the input interval graph. But, does every
interval graph have a universal representation? Fortunately, we answer this question affirmatively by giving a constructive proof (the
details can be found in the proof of the
\cref{prop:special}). Hence, we conclude that the universal interval representation is a good proxy for studying the complexity of interval
graphs, and as a consequence, throughout this paper, we crucially adopt the convention that our input interval graphs are always represented
by adhering to this protocol.

Since each interval $i$ can have $n-i+1$ different right endpoints, i.e., $[i,i], \dots, [i, n]$, there exists up to $n!$ distinct universal
interval representations,
which intuitively implies that the information-theoretic minimum of the universal interval
representation is $\log (n!)$ bits.
Furthermore, any interval $I_i$ can be decoded using $\adj{}$ queries only by finding $\min_{j > i, \adj(i, j) = \text{false}} j$. Also,
$I_i$ can be decoded using
$\deg{}$ queries only, by decoding from the leftmost interval (the detailed procedure is described in \cref{sec:upper}).
Therefore, any data structure which supports $\adj{}$ or $\deg{}$ queries on interval graphs in universal interval representation should
consume at least $\log (n!)$ ($\approx n \log n-1.44n$) bits  of  space.
As alluded to earlier, the state-of-the-art information-theoretic lower bound to
encode any interval graph $G$ on $n$ nodes is given
by $n\log n - 2n\log \log  n - O(n)$ bits~\cite{DBLP:journals/algorithmica/AcanCJS21}.\footnote{Recently \textcite{bukh2022enumeration}
    claim an improved (tight) bound of $\approx n \log n - 2.478n + \Theta(\log n)$ bits.} As our lower bound does not deviate too much
from the information-theoretic minimum to encode unlabeled interval graphs, we believe that the notion of universal interval representation
is a very natural choice while studying the complexity of interval graph data structure.
We also observe that comparing with the above implementation of~\cite{DBLP:journals/algorithmica/AcanCJS21}, our universal interval
representation offers two salient benefits,
\begin{enumerate*}[label=(\roman*)]
    \item the implementation becomes more space-efficient, i.e., the benchmark (the information-theoretic minimum of the input) becomes
          $\log(n!)\approx n\log n-1.44n$ instead of $\approx n \log n -0.44n$ bits, saving $n$ bits of space implicitly, and
    \item as we will see in \cref{sec:tradeoff} and \cref{sec:upper}, it would be easier to analyze the distribution of interval graphs in
          universal interval representation for providing lower and upper bounds, since the left  endpoints are always fixed, whereas
          in~\cite{DBLP:journals/algorithmica/AcanCJS21}, left endpoints could be arbitrary distinct points from $[2n-1]$.
\end{enumerate*}

In conclusion, we believe that the assumption of the universal interval  representation is very natural, practical, and advantageous when
applied to interval graphs, and for this reason, it is also our method of choice in this work. Now we are ready to state the main results of
our paper.

\subsection{Our Contribution}\label{sec:contribution}
Since Jacobson's pioneering work~\cite{Jacobson89} on succinct data structures for trees, extensive research has led to substantial progress
in achieving such data  structures for various other graph classes.
In the regime of succinct space, the time cost of data structures very often is remarkably higher than in the regime of linear space, however
showing such a separation remains a challenge since~\cite{Jacobson89}.
Note that tight lower bounds  have been shown only for a handful of problems such as  \texttt{RANK}~\cite{DBLP:conf/soda/PatrascuV10}, range
minimum query~\cite{liu2021nearly,DBLP:conf/stoc/LiuY20}, permutations~\cite{DBLP:conf/soda/Golynski09}, systematic data structure for
boolean matrix-vector multiplication~\cite{10.1145/3188745.3188830} and sampler~\cite{10.1145/2488608.2488707} in the  succinct  regime. Our
research adds to this growing body of work by providing novel time/space trade-off lower bounds for succinct data structures designed for
interval graphs. More
specifically, we start by showing the following impossibility result, which, informally speaking, states that any data structure for interval
graphs that
assumes that the input interval graph is represented using a universal interval representation must have a high lower bound if the data
structure aims to
support both $\adj$ and $\deg$ query together.

\begin{restatable}[]{theorem}{thmtradeoff}\label{thm:tradeoff}
    Any data structure taking $\log (n!) + r$ bits for representing an input interval graph with $n$ vertices
    that answers $\adj$ in $t_1$ time, and $\deg$
    in $t_2$ satisfies
    \[
        r \ge \frac{n}{(\log n)^{O(t_1+t_2)}}
    \]
    in the cell-probe model with $w = \Theta(\log n)$ bits word size when the input interval graph is provided using a universal interval
    representation.
\end{restatable}

Surprisingly, there exists a dichotomy between these two queries which we summarize in the following two theorems. Here we assume the usual
model of computation, i.e.,
a $\Theta (\log n)$-bit word RAM model while designing the following succinct data structures.

\begin{restatable}[]{theorem}{thmimprovedadj}\label{thm:improved_adj}
    There exists a data structure that supports $\adj$  queries  in $O(1)$ time using $\log(n!) + O(\sqrt{n}\log n)$ bits for $n$-vertex
    interval graphs represented in universal interval representation.
\end{restatable}

\begin{restatable}[]{theorem}{thmimproveddeg}\label{thm:improved_deg}
    There exists a data structure that supports $\deg$  queries  in $O(1)$ time using $\log(n!) + O({n}^{2/3}\log n)$ bits for $n$-vertex
    interval graphs represented in universal interval representation.
\end{restatable}

Notice that for constant time $\adj$ and $\deg$ queries together, our lower bound from \cref{thm:tradeoff} results in
$\log(n!) + O(n/\polylog(n))$ bit, thus, from the
perspective of succinct data structure complexity, we capture an interesting phenomenon of incompatibility between two different queries on
the same database: There are
efficient data structures for the two queries individually, but any data structure which supports two queries simultaneously has a high lower
bound.
To the best of our knowledge, this is the first proof of such a phenomenon.
Previously, \textcite{DBLP:conf/soda/Golynski09}
proved lower bounds for three different problems with multiple queries.
It turns out that two of the problems (text searching \& retrieving, rank/select query on matrix) can be reduced from the rank/select problem
on a boolean string, which has high lower
bounds with individual query by~\cite{DBLP:conf/soda/PatrascuV10}.
On the other hand, there is no known efficient succinct data structure (e.g., constant time and $n^{0.9}$ bits of
redundancy)
for the third problem (permutation) with an individual query, to our knowledge.

\subsection{Previous techniques}\label{sec:previoustech}
The state-of-the-art lower bound techniques in cell-probe model are derived from~\cite{DBLP:journals/jcss/MiltersenNSW98}.
In this seminal work, Miltersen et al.\ propose two important techniques, \emph{richness} and \emph{round elimination}.

The richness technique is refined and improved later on by~\cite{patrascu2006higher,PTW10}.
The technique appears in context of proving lower bound for static data structure~\cite{10.1109/SFCS.1989.63450,PTW10,Larsen12b,GL16,Yin16},
dynamic data
structure~\cite{Larsen12a,CGL15,LWY18}, streaming~\cite{LNN15}, and succinct data
structure~\cite{GM07,10.1145/2488608.2488707, 10.1145/3188745.3188830}.
The richness technique is similar to the rectangle refutation technique in the literature on communication complexity.
The richness technique is all about showing that
\begin{enumerate*}[label=(\roman*)]
    \item a too-good-to-be-true data structure implies the existence of a large monochromatic (or ``almost monochromatic'' for randomized
          data structures) combinatorial
          rectangle in the answer table of the data structure problem;
    \item large (``almost'') monochromatic rectangles do not exist in the answer table.
\end{enumerate*}
Specifically, we find a large set $Q$ of queries and a large set $B$ of databases by observing a too-good-to-be-true data structure, and
proving that
\begin{enumerate*}[label=(\roman*)]
    \item for any $q\in Q, b\in B$, the answer to $q$ on $b$ is the same (or ``almost the same'' for randomized data structure);
    \item by extremal combinatorics, in the answer table of the data structure problem, there do not exist such $Q$ and $B$.
\end{enumerate*}
It is known that the technique works well in proving lower bound for problems which satisfy that the answers to a few fixed queries reveal
much information about the database.

The round elimination technique is also used in, and refined to adapt to, various scenarios accordingly, including static data structure with
polynomial
space~\cite{SEN2008364,chakrabarti2010optimal,10.1145/3209884}, static data structure with near-linear
space~\cite{patracscu2006time,DBLP:conf/soda/PatrascuT07}, and succinct data
structure~\cite{DBLP:conf/soda/PatrascuV10,DBLP:conf/stoc/LiuY20,liu2021nearly}.
The technique works in the following manner:
\begin{enumerate*}[label=(\roman*)]
    \item suppose there is a data structure with time cost $t$, then we can modify the data structure to obtain a new data structure that has
          a time cost of $t-1$ but
          consumes more space, and/or works only for a smaller database;
    \item we recursively apply this modification and at the end of the recursion, there is a data structure that solves a non-trivial problem
          but does not access any memory
          cell, so we obtain a contradiction.
\end{enumerate*}
In this work, we adapt the ideas from~\cite{DBLP:conf/soda/PatrascuV10,DBLP:conf/stoc/LiuY20} to prove the lower bound.

\Textcite{DBLP:conf/soda/Golynski09} develops an independent technique to prove lower bounds for a special class of problems that support a
pair of opposite queries
which can ``verify''
each other, for example, $\pi(i)=j$ and $\pi^{-1}(j)=i$ assuming the database is a permutation $\pi$.
The central observation is the following.
Assume the time cost is low, then there is a cell $d$ such that the number of opposite query pairs which access $d$ simultaneously is only
$\beta\ll w/\log w$ where the
word size is $w$ bits.
Thus, we can compress the data structure by removing $d$ in the following way:
\begin{enumerate*}[label=(\roman*)]
    \item write down the address of $d$ to identify $d$;
    \item enumerate and execute queries to identify the $\beta$ opposite query pairs which access $d$ simultaneously;
    \item write down the bijection between the queries with $\beta\log \beta$ bits.
\end{enumerate*}
Indeed, for a pair of opposite queries $\pi(i)=j$ and $\pi^{-1}(j)=i$, if only $\pi(i)$ accesses $d$, then $\pi^{-1}(j)=i$ can be answered
without knowing $d$.

\Textcite{natarajanramamoorthy_et_al:LIPIcs:2018:8862} propose a different technique to prove lower bound for dynamic non-adaptive succinct
data structures with the
sunflower lemma.
They proved lower bounds for a couple of elementary data structure problems, including returning the median (or the minimum) of a set of
numbers, and predecessor search.
Their technique essentially leverages the \emph{non-adaptivity} of the \emph{dynamic} data structures. Non-adaptivity means that the set of
cells which are accessed by the query/update algorithm is independent of the database.
Therefore, the technique is not applicable in our case, as we study \emph{fully-adaptive} \emph{static} data  structures.

Very recently, \textcite{li2023tight} proved a tight lower bound for dynamic succinct dictionaries.
They extend a classical idea in proving lower bound for the dynamic data structures, and develop an ad-hoc analysis for dynamic dictionary
problem.
Their technique intrinsically utilizes the update operations of \emph{dynamic} data structure, and hence does not apply to our case of
\emph{static} data structure.

\section{Technical Overview}\label{sec:techover}
Our lower bound results presented in \cref{sec:tradeoff} are established in the \emph{cell-probe} model~\cite{DBLP:journals/jacm/Yao81a}.
In the cell-probe model, we have given two algorithms, $C,Q$ which can have unlimited computational power.
$C$ takes a certain structure, such as a graph, and produces a table of $s$ words with each word having $w$ bits
(the word size).
Given any query, we access the table using the query algorithm $Q$.
We are allowed to adaptively access at most $t$ cells of the table.
Meaning, given a query, $Q$ decides the first cell to probe. Then $Q$ decides using the cell's
content the next cell to probe and so on. At the end, it accessed $t$ cells and outputs the
answer to the query according to the contents of the $t$ cells.
We call the tuple of parameters $(s,w,t)$ the \emph{cell-probe complexity} of the algorithm.
This is the only complexity metric we will use for our lower bounds.
Additionally, we augment the model as in~\cite{DBLP:conf/soda/PatrascuV10} with $P$ extra bits which $Q$ is allowed to access for free.
The bits are called \emph{published bits}.
Informally, one may consider the published bits as the cache of the CPU, so it is (almost) free to access.

\subsection{The Issue of Applying \texorpdfstring{\patrascu{}}{Patrascu} and Viola's Technique}\label{sec:pv10tec}
We adopt the basic framework from~\cite{DBLP:conf/soda/PatrascuV10}.
As we discussed previously, the framework is called round elimination.
Without loss of generality, suppose the information-theoretic lower bound of the database is $n$ bits, the word size is $w$ bits, and we have
a data structure that takes $n+r$ bits of space and solves any query in $t$ time.
Besides the memory,~\cite{DBLP:conf/soda/PatrascuV10} further assumes that there are some \emph{published bits}, which the query algorithm
can access \emph{for free}.
Given a succinct data structure using $n+r$ bits of space and $t$ time, one may construct a data structure that consumes $n$ bits of space,
$r$ bits of published bits and
solves any query in $t$ time.
The round elimination in~\cite{DBLP:conf/soda/PatrascuV10} then works in the following manner:
\begin{enumerate*}[label=(\roman{*})]
    \item Given a data structure using $p$ bits of published bits and $t$ time, we can find a set $\mathcal{P}$ of memory cells such that a
          fraction of queries, e.g., a constant fraction, access at least one cell in $\mathcal{P}$.
    \item Then we copy the contents together with the addresses of the cells in $\mathcal{P}$ to the published bits, and modify the original
          query algorithm so that whenever the
          original query algorithm wants to access a cell in $\mathcal{P}$, the query algorithm can save the memory access with the
          information we copied to the published bits.
    \item Thus, we obtain a data structure which uses $p+|\mathcal{P}|(\log n+w)$ published bits and answers an average query with
          $t-\Omega(1)$ times of memory access.
    \item We recursively apply this modification and at the end, we obtain a data structure which does not access memory at all but can solve
          any query. Then the final data structure must use at least $n$ published bits.
\end{enumerate*}
In~\cite{DBLP:conf/soda/PatrascuV10}, $|\mathcal{P}|$ is always in $\Theta(p\log n)$.
Hence, at the end of the recursion, the authors have $r(w+\log n)^{O(t)}$ published bits, and they can prove $r=n/(\log n)^{O(t)}$ by
assuming $w=O(\log n)$.

The issue of applying the original technique in~\cite{DBLP:conf/soda/PatrascuV10} to our problem happens in the very first step.
Intuitively,~\cite{DBLP:conf/soda/PatrascuV10} can find such $\mathcal{P}$ because the answers to the queries are correlated and the database
is encoded succinctly. Hence, many queries
must access the same cells to retrieve the same information.
For the interval graph, the answers to adjacency queries are almost independent of each other\footnote{
    One may think that we are dealing with an array $A$ of integers from $[n]\times[n-1]\dots[2]\times[1]$ which is the length of eac
    interval.
    To answer $\adj(i,j)$ with $j\ge i$, we simply retrieve $A_i$ and
    check if $A_i\ge j-i+1$.
    So the answers to queries are basically independent of each other.
    In fact, if one would like to solve adjacency query only, in the cell-probe model we can have a data structure using $\log(n!)+1$ bits of
    space and constant time by
    generalizing the technique in~\cite{DBLP:conf/stoc/DodisPT10}.
    See formal proof in \cref{sec:cell_adj}.
}.
A similar thing holds for the degrees of the vertices.

\pagebreak
\paragraph{Observations:}
\begin{enumerate}[label=(\arabic*)]
    \item Answers to adjacency queries can verify the length of each interval.
    \item For each interval with left endpoint $i$, if one can know $\deg(i)$ together with the right endpoint of the interval
          (denoted by $e_i$), one can know the number of right endpoints which land before position $i$ (denoted by $R(i)$) as $e_i=\deg(i)+R(i)$.
    \item \cite{DBLP:conf/soda/PatrascuV10} proves a lower bound for the \texttt{RANK} problem where the database is a Boolean string
          $\{0,1\}^n$ and the answer to query $i$ is the number of ones before position $i$.
\end{enumerate}
Hence, our problem is somehow similar to the one in~\cite{DBLP:conf/soda/PatrascuV10}.
To leverage this observation, we will exploit that ``answers to adjacency queries can verify the length of each interval''.
At first glance, this might be a problem as
querying $\adj(i,j)$ for $j\ge i$ only returns one bit to tell us if the length of interval $i$ is at least $j-i$. There is no fixed query to
reveal the length.
Therefore, we focus on a tuple of \emph{hard} queries $\deg(i)$, $\adj(i,e_i)$, and $\adj(i,e_i+1)$.
Notice that $e_i$ is hidden from us and is a random variable which depends on the database.
The situation reminds us of~\cite{DBLP:conf/stoc/LiuY20} which fixes a similar issue about proving lower bound for
some hard queries which depend on the database.

\subsection{Towards applying \texorpdfstring{\citeauthor{DBLP:conf/stoc/LiuY20}}{Liu and Yu}'s technique}\label{sec:LYtech}
To find the set $\mathcal{P}$ of cells such that many hard queries access at least one cell in $\mathcal{P}$, we should find a set $Q$ of
queries such that the answers to
$Q$
are highly
correlated with the hard queries and the \emph{answers} to the hard queries. Then we let $\mathcal{P}$ be the set of cells which are accessed
by $Q$.
Previous works~\cite{DBLP:conf/soda/PatrascuV10,DBLP:conf/stoc/LiuY20,liu2021nearly} apply a proof by contradiction to show that many hard
queries access at least one cell in
$\mathcal{P}$.
\Cite{DBLP:conf/soda/PatrascuV10} shows that if $\mathcal{P}$ is barely accessed by hard queries, assuming the space cost is $n+p$ bits. Then
we can apply the following encoding:
\begin{enumerate}
    \item Write down the content of $\mathcal{P}$, which is the cells accessed by $Q$, in the order of being accessed by the query algorithm.
    \item The amount of information shared by the answers to hard queries and the answers to $Q$ is $\omega(p)$ bits.
          Therefore, we write down the content of the cells accessed by the hard queries, given the content of $\mathcal{P}$, with the
          optimal encoding to save $\omega(p)$ bits
          of space.
    \item Write down the contents of the remaining cells in the lexicographical order of their addresses.
\end{enumerate}
This allows us to encode the data structure with $n-\omega(p)$ bits which is impossible.

When both $Q$ and the hard queries are fixed and independent of the database,
we only have to write down the content of $\mathcal{P}$ and the cells accessed by hard queries in the order in which they are accessed by the
query algorithm, as the address of the next
cell is automatically revealed by simulating the query algorithm on previous cells.
When $Q$ or the hard queries are determined by the database, the simple argument in~\cite{DBLP:conf/soda/PatrascuV10} does not work anymore
because we cannot find the set
of cells we are supposed to compress to save space.
In fact, the information which is shared by the two sets of cells, i.e., the information we are supposed to compress, may precisely be
``what is the set of cells'', so we
cannot simply write the addresses down.

\Cite{DBLP:conf/stoc/LiuY20,liu2021nearly} overcome this.
They note that if we can encode $Q$ at low cost, then we can find an $\ell\in [t]$ such that the contents of the cells, which are accessed by
$Q$ in $\ell$-th round of
memory access, share a lot of information with the answers to the hard queries.
The idea works because the addresses of the $\ell$-th cells are revealed by the contents of the previous $\ell-1$ cells.
They introduce a new encoding argument to use this:
\begin{enumerate}
    \item Write down a few bits to encode $Q$ together with $\ell$.
    \item Write down the contents of the cells which are accessed by $Q$ in the first $\ell-1$ rounds of
          memory access in the order in which they are accessed by the query algorithm.
    \item Write down the contents of the remaining cells, except the cells which are accessed by $Q$ in the $\ell$-th round of memory access,
          in the lexicographical order of their addresses. Note that most of the cells which are accessed by the hard queries are written down in
          this step because they barely intersect with the $\ell$-th cells by the assumption.
    \item Enumerate every query and try to answer them by simulating the query algorithm on the cells we have written down.
          We have enumerated many hard queries and revealed answers to them because of the disjointedness given
          by the assumption, and we know a lot of information about the contents of the
          missing cells because of the aforesaid correlation.
          Then we write down the contents of the missing cells with the known information and the optimal encoding to save $\omega(p)$ bits
          of space.
\end{enumerate}

\paragraph{$Q$ depends on the dataset.}
An immediate issue arises when one tries to apply~\cite{DBLP:conf/stoc/LiuY20,liu2021nearly}'s technique on the hard query tuples
$(\deg(i),\adj(i,e_i),\adj(i,e_i+1))$.
Recall that we would like to utilize $R(i)$, the number of right endpoints that land before position $i$.
It turns out that the only query, which is correlated with $R(i)$'s and which we could find, is another tuple of hard queries.
In other words, $Q$ cannot be encoded at low cost in our case.
We observe that $e_i$ can be computed by a binary search using $\adj(i,\cdot)$ queries.
Hence, we let all the cells accessed by the binary search procedure be the published bits.
One may think that $Q$ is fixed in some sense, but we modify the query algorithm so that $Q$ probes a lot of cells, so that the answers to
$Q$ reveal a lot of information
about the hard queries and their answers.
Specifically, in the original proof, $Q$ consists of $\approx p$ queries, each access $t\le\log n$ cells adaptively according to the query
algorithm, forming $\approx p$
sequences of cells of length $t\le\log n$. In our case, $Q$ still  consists of $\approx p$ queries but each access $t\log n\le \log^2 n$
cells adaptively. This forms $\approx p$
sequences of cells of length $t\log n\le\log^2n$ and each sequence of cells is obtained by adaptively combining the original query algorithm.
The approach is a natural extension of the idea of~\cite{DBLP:conf/stoc/LiuY20,liu2021nearly}.

\paragraph{Mutual information.}
Another technical issue is to lower bound the mutual information between the answers to $Q$ and the answers to hard queries.
The mutual information is a significant ingredient in Step 4 of the aforesaid new encoding argument.
The distribution here is very different from the previous ones and is sophisticated, hence, analyzing it is a highly non-trivial job.
Indeed, in the problem studied by~\cite{DBLP:conf/soda/PatrascuV10}, in each position the number of ones has a Bernoulli distribution with
parameter $1/2$ and the numbers
of ones in different positions are mutually independent of each other. Therefore, the number of ones in any range has binomial distribution
which we can easily analyze
with standard techniques.
In our case, the ``number of right endpoints in position $i$'' has Poisson
binomial distribution with parameters $p_1=1/n,p_2=1/(n-1), \dots, p_i=1/(n-i+1)$ and the numbers in different positions  are correlated with
each other.

We still adopt the idea from~\cite{DBLP:conf/soda/PatrascuV10} that breaking the universe of queries
into $p$ consecutive blocks of length $n/p$, denoted by $q_1,\dots,q_p$, and
letting $Q$ together with answers to $Q$ reveal the number of right endpoints in each block.
Assuming only a few hard queries access some cell accessed by $Q$, then by an average argument, there
is a set of hard queries which break the universe into $p$ consecutive
blocks, denoted by $q'_1,\dots,q'_{p}$, of length $n/p$, with the first and the last blocks being exceptions we ignore for the overview.
Additionally, the intersection of $q_{i+1}$ and $q'_{i}$ is of length $(1-o(1))n/p$ for
any $i$.
These are now two different partitions of the same universe $[n]$.
In order to measure the information contained in the numbers of right endpoints in the blocks, we set a threshold for each of the $2p$ blocks
of length $n/p$ in the two
partitions and check if the number of right endpoints is larger than the threshold for each aforesaid block.
We use random variable indicators $T_i$'s and $T'_i$'s to denote these.
The proof for lower bounding mutual information then consists of two components: Showing that the entropy of $\{T_i\}$ is $\Omega(p)$, and
showing that the entropy of $\{T_i\}_{i\in [p]}$ is  $o(p)$ given $\{T'_i\}_{i\in [p]}$.

We achieve the first component by the following argument.
Let $Y_i$ denote the number of right endpoints before block $q_i$.
To lower bound $H(T_1,\dots,T_p)$, it suffices to lower bound $H(T_i \mid Y_i)$ by the chain rule and
observing that $H(T_i|T_1,\dots,T_{i-1})\ge H(T_i|Y_1,\dots,Y_i)=H(T_i|Y_i)$.
Let $Y'_i$ denote the number of right endpoints which are in block $q_i$ but their corresponding left endpoints are not. Let $Z_i$ denote the
number of intervals which are fully contained
in block $q_i$. Then $T_i$ indicates if $Y'_i+Z_i$ is large enough.
We further observe that $Y'_i$ has a binomial distribution as long as $Y_i$ is fixed, and $Y_i$ concentrates around its expectation with high
probability.
Thus, we are able to prove that $Y'_i$ acts like a binomial random variable and, hence, deviates from some fixed value in an approximately
symmetric way, and that
$|Z_i-\E{[Z_i]}|$ is dominated by the former deviation with high probability.

To upper bound $H(T_1,\dots|T'_1,\dots)$, it suffices to upper bound $H(T_{i+1}|T'_i)$ by the chain rule and the sub-additivity of entropy.
Note that the range $q_{i+1}\cup q'_i$ consists of a common part of length $(1-o(1))n/p$ and two distinct parts of length $o(n/p)$.
We are going to prove that the number of right endpoints in the common part acts like a binomial random variable with a similar proof to the
one in the previous  paragraph.
So it deviates from its expectation like a binomial distribution.
For the ``numbers of the right endpoints'' in the distinct parts, since they are sums of a binomial random variable and a Poisson binomial
random variable given fixed $Y_i$, and since distinct parts are of length $o(n/p)$, we are able to prove that their variances are small.
This allows us to apply Chebyshev's inequality to upper-bound the deviations.

\subsection{Upper Bound Techniques}\label{sec:uppertech}
\Textcite{DBLP:journals/algorithmica/AcanCJS21} recently provided an $(n \log n + 2n)$-bit representation of an interval graph $G$ as
follows: They first represent
$G$ as $n$ intervals of $2n$ distinct endpoints from $[2n]$. Here, each vertex is denoted as $i$ if and only if the corresponding interval's
left endpoint is at the
$i$-th leftmost position among all $n$ left endpoints.
The encoding consists of two components:
\begin{enumerate*}[label=(\roman*)]
    \item A binary sequence of size $2n$ that indicates whether a given point $i$ is the leftmost or rightmost endpoint, and
    \item an integer sequence of size $n$ that stores the right endpoints according to the order of their leftmost points.
\end{enumerate*} Using this representation, one can decode the interval $I_i$
corresponding to vertex $i$.
In addition to this, their data structure provides efficient support for both $\adj{}$ and $\deg{}$ queries in $O(1)$ time by combining the
representation with an
$o(n)$-bit auxiliary structures for supporting \texttt{RANK} and \texttt{SELECT} queries in $O(1)$ time~\cite{ClarkM96}.

In order to improve this result by  $\Omega(n)$  bits while supporting individual queries in $O(1)$ time, we introduce mainly two ideas:
\begin{enumerate*}[label=(\roman*)]
    \item As alluded previously, we assume the input graph is in the universal interval representation, so that the entropy of the input
          database is $\log(n!)\approx n\log n-1.44n$ bits instead of $\approx n \log n -0.44n$ bits, and
    \item we use a fairly standard trick of dividing the input into smaller-sized sub-blocks to encode them even more space-efficiently
          without compromising the query time.
\end{enumerate*}
The previous data structure works with the global information about the input directly without exploiting the interaction between the
intervals in a more granular
fashion.
We provide a high-level overview below whereas all the details can be found in \cref{sec:upper}.

\paragraph{Adjacency query:}
Recall that we assume that the graph is represented in universal interval representation.
To answer adjacency query $\adj(i,j)$ with $i\le j$, it suffices to retrieve the length of interval $i$ and then compare the right endpoint
with $j$.
Hence, it suffices to encode an array $A\in [n]\times[n-1]\cdots[2]\times[1]$ which reminds us of the technique
in~\cite{DBLP:conf/stoc/DodisPT10}.
This technique works for encoding an
array from $[\Sigma]^n$ with $n\log\Sigma+O(\log n)$ bits.
In fact, in the cell-probe model, i.e., assuming the CPU has unlimited computational power, the array $A$ can be encoded in $\log(n!)+1$ bits
by
generalizing~\cite{DBLP:conf/stoc/DodisPT10}. We exhibit this in \cref{sec:cell_adj}.
To apply the technique from~\cite{DBLP:conf/stoc/DodisPT10} in the word RAM model, we divide the array into $\sqrt n$ blocks of equal length.
To encode a block, which is a sub-array from $[n-k\sqrt n] \times\cdots\times [n-(k+1)\sqrt n+1]$ for some $k\in\{0,\dots,\sqrt n-1\}$,
we simply apply~\cite{DBLP:conf/stoc/DodisPT10}'s encoding scheme with a universe of size $[n-k\sqrt n]^{\sqrt n}$.

\paragraph{Degree query:}
\begin{figure}
    \begin{tikzpicture}
        \tikzstyle{Ii} =[fill, thick, black]
        \tikzstyle{i}  =[fill, thick, red];
        \tikzstyle{ii} =[fill, thick, gray];
        \tikzstyle{iii}=[fill, thick, blue];
        \tikzstyle{iv} =[fill, thick, violet];

        \draw[draw=black] (0,0) rectangle ++(3,2);
        \draw[draw=black] (3,0) rectangle ++(3,2);
        \draw[draw=black] (6,0) rectangle ++(3,2);

        \draw[Ii] (4.5,0.2) rectangle (8.5,0.15);
        \draw[dashed] (4.5,0) -- (4.5,2);

        \draw[i] (1,0.4) rectangle (8,0.45);
        \draw[i] (2,0.6) rectangle (7,0.65);

        \draw[ii] (5,0.8) rectangle (7.5,0.85);
        \draw[ii] (4.75,1.0) rectangle (8.75,1.05);

        \draw[iii] (0.5,1.2) rectangle (4,1.25);
        \draw[iii] (1.5,1.4) rectangle (5,1.45);

        \draw[iv] (3.5,1.6) rectangle (4.25,1.65);
        \draw[iv] (3.25,1.8) rectangle (8,1.85);

        \node[Ii] at (-3,2) {};
        \node at (-1.75,2) {$i$-th interval};
        \node[i] at (-3,1.5) {};
        \node at (-2.5,1.5) {(i)};
        \node[ii] at (-3,1) {};
        \node at (-2.5,1) {(ii)};
        \node[iii] at (-3,0.5) {};
        \node at (-2.5,0.5) {(iii)};
        \node[iv] at (-3,0) {};
        \node at (-2.5,0) {(iv)};
    \end{tikzpicture}
    \caption{Type of Intervals for Degree Query\label{fig:upper_bound_intervals}}
\end{figure}
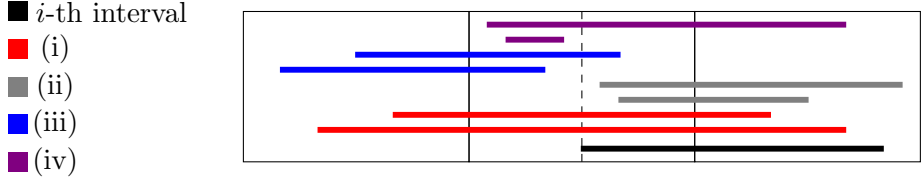

To answer degree query, we again apply the trick of block partitioning but with the block size of $\sqrt[3]n$.
For an interval $i$ in the $k$-th block, there are four different intervals which may intersect with interval $i$
(See \cref{fig:upper_bound_intervals}).
\begin{enumerate*}[label=(\roman*)]
    \item The left endpoint lands before the $k$-th block, and the right endpoint lands after the $k$-th block;
    \item the left endpoint lands in the interval $i$;
    \item the left endpoint lands before the $k$-th block, and the right endpoint lands in the $k$-th block;
    \item the left endpoint lands in the  $k$-th block but before the position $i$.
\end{enumerate*}
For each interval in the same block, the number of intervals of kind (i) is the same, so we can store them with an array using
$O({n}^{2/3}\log n)$ bits.
Assuming there are $b_k$ right endpoints land in block $k$, the numbers of the other three kinds of intervals are at most
\begin{enumerate*}[label=(\roman{*})]\setcounter{enumi}{1}
    \item $n-i$,
    \item $b_k$,
    \item $i-(k-1)\sqrt[3] n$,
\end{enumerate*}
respectively.
Therefore, it suffices to encode an array with universe
$[n+b_1]^{\sqrt[3] n}\times\cdots \times [n-(k-1)\sqrt[3] n +b_k]^{\sqrt[3] n}\times\cdots$ with constraint $\sum_k b_k=n$.
Finally, we apply~\cite{DBLP:conf/stoc/DodisPT10} and calculate the space cost carefully to complete the proof.

\section{Preliminaries}\label{sec:defs}
In this section, we state a couple of basic preliminaries which we need for the upper and lower bounds. First, we want to show that it is
possible to transform any interval
graph into our universal interval representation.
\begin{proposition}\label{prop:special}
    Any interval graph $G$ of $n$ vertices has a universal interval representation.
\end{proposition}
\begin{proof}
    According to \textcite{Hanlon}, $G$ has a well-known interval representation $I$,
    where all the endpoints in $I$ are distinct integers from the set $[2n]$. We can then
    generate a universal interval representation of $G$ using the following procedure:

    Starting from the point $1$, we check each points in $[2n]$ from left to right. If $i$ is the left endpoint of the interval in $I$, we
    check the next point. Otherwise,
    let $j < i$
    be the rightmost left endpoint to the left of $i$. Then we modify $I$ in two ways. First, we modify all intervals in the set
    $\{[p, q] \in I \mid p < i, q \ge i\}$ to $[p, q-(j-i)]$. Second, we modify all intervals in the set $\{[p, q] \in I \mid\ p > i \text{~and~} q > i\}$
    to $[p-(j-i), q-(j-i)]$. We repeat
    this modification process until $I$ has an interval with a left endpoint of $i$ or there is no endpoint to the right of $i-1$. Since this
    modification process does not
    alter the intersection relations among the intervals in $I$, and preserves the uniqueness of the left endpoints of the intervals in $I$,
    we can obtain a universal interval
    representation of $G$ after completing the procedure.
\end{proof}

As mentioned previously, we adopt the convention that any
interval graph will be presented in its universal interval representation throughout the remainder of this paper. Additionally, sampling an
interval graph from its universal
interval representation is straightforward, and we make use of the following
two essential properties crucially in our lower-bound proofs.
\begin{proposition}\label{prop:independent}
    Let the universal interval representation be drawn uniformly at random. Then for each $i \in [n]$,
    \begin{enumerate*}[label=(\roman*)]
        \item $e_i$ is uniformly distributed over the range $[i, n]$, and
        \item $e_1, e_2, \dots, e_n$ are mutually independent.
    \end{enumerate*}
\end{proposition}

\paragraph{Terminologies for Cell-probe Lower Bounds}
Recall that any query can be answered in $t$ adaptive memory accesses by assumption, and that the word size is $w$ bits.
We adopt the notation from~\cite{DBLP:conf/soda/PatrascuV10,DBLP:conf/stoc/LiuY20} and define the following
for a query $q$ or a query set $Q$ on the input data:
\begin{itemize}
    \item For a query $q$ and an integer $\ell\in[t]$, the address of the cell probed by $q$ in $\ell$-th memory access is denoted by
          $\Probe{}_\ell(q)$.
          We also define $\Probe{}_\ell(Q)\triangleq\cup_{q\in Q}\Probe{}_\ell(q)$, $\Probe{}_{\le \ell}(q)\triangleq\cup_{i\le \ell}\Probe{}_i(q)$,
          $\Probe{}_{< \ell}(q)\triangleq\cup_{i< \ell}\Probe{}_i(q)$ and
          $\Probe{}_{\le\ell}(Q)\triangleq\cup_{q\in Q}\Probe{}_{\le\ell}(q)$, $\Probe{}_{< \ell}(Q)\triangleq\cup_{i< \ell}\Probe{}_i(Q)$.
    \item The set of cells probed when answering a query $q$ (resp.\ all queries in $Q$) is denoted by
          $\Probe{}(q)\triangleq\cup_{i\in[t]}\Probe{}_i(q)$
          (resp.\ $\Probe{}(Q)\triangleq \cup_{q\in Q}\Probe{}(q)$).
    \item A binary string encoding the content of the $\ell$-th probed cell when answering the query $q$ is denoted by
          $\Foot{}_{\ell}(q)\in\{0,1\}^w$. We also define
          $\Foot{}_{\ell}(Q)\triangleq\Foot{}_{\ell}(q_1)\cdot\Foot{}_{\ell}(q_2)\dotsb\in(\{0,1\}^w)^{|Q|}$ as the concatenation of all $\Foot{}_{\ell}(q)$ for $q \in Q$ in
          lexicographical order $q_1<q_2<\dotsb$.
    \item The binary string obtained by concatenating the contents of the probed cells when answering a query $q$ is called its
          \emph{footprint}, denoted by
          $\Foot{}(q)\triangleq \Foot{}_{1}(q)\cdot\Foot{}_{2}(q)\dotsb\in(\{0,1\}^w)^t$. We also define $\Foot{}(Q)\triangleq\Foot{}(q_1)\cdot\Foot{}(q_2)\dotsb$ as the concatenation of
          the footprints of all queries in $Q$ in lexicographical order $q_1<q_2<\dotsb$.
    \item $\Foot{}_{<\ell}(q)\triangleq \Foot{}_{1}(q)\cdot\Foot{}_{2}(q) \cdots \Foot{}_{\ell-1}(q)$ is obtained by concatenating the
          contents of the first probed cells up to the $(\ell-1)$-th probe. We also define
          $\Foot{}_{<\ell}(Q)\triangleq\Foot{}_{<\ell}(q_1)\cdot\Foot{}_{<\ell}(q_2)\dotsb$
          as the concatenation of all $\Foot{}_{<\ell}$ for all queries in $Q$ in lexicographical order $q_1<q_2<\dotsb$.
\end{itemize}

Note that the $\ell$-th probe of any query algorithm
can be determined based on the first $\ell-1$ probes, as the algorithm can only act based on the known information.
In other words, $\Probe{}_\ell(q)$ (resp.\ $\Probe{}_\ell(Q)$) is known given $\Foot{}_{<\ell}(q)$ (resp.\ $\Foot{}_{<\ell}(Q)$) for any
$\ell$ and $q$ (resp.\ $Q$). This
implies that one can decode the address of the next cell to be probed using $\Foot{}_{<\ell}(q)$.
Therefore, one can answer the query $q$ using its footprint $\Foot{}(q)$ as follows:
At each step, one checks whether the cell to be probed has already been published or not. If the cell has been published (hence, one knows
the contents of the cell), one
can simply read its contents without incurring any additional cost and move on to the next step. Otherwise, one needs to read the next $w$
bits of $\Foot{}(q)$, which
contains the contents of the cell to be published as well as the address of the next cell to read. As a result, the length of $\Foot{}(q)$ is
exactly $w \cdot \lvert \Probe{}(q)\rvert$ bits.

\paragraph{Useful lemmas}
Proving our upper bounds uses the following well-known lemma that gives us a succinct encoding
for integer sequence of arbitrary alphabet size without losing the
constant time accessibility of the data structure.

\begin{lemma}[\cite{DBLP:conf/stoc/DodisPT10}]\label{lem:DPT}
    For any integer sequence of size $n$ with values from some
    alphabet $\Sigma = \{0, 1, \dots, |\Sigma|\}$ can be stored using $\ceil {n\log \lvert\Sigma\rvert} + O(\log n)$
    bits while supporting $O(1)$-time access to any position of the sequence in the word RAM model.
\end{lemma}

Finally, we state some well-known probability theory propositions for completeness.
\begin{proposition}[Chebyshev's inequality]
    Let $X$ be a random variable with finite mean $\E[X]$ and variance $\sigma^2$. Then for any $k >0$,
    \[
        \Pr[X-\E[X]\leq k \sigma] \leq 1/k^2.
    \]
\end{proposition}

\begin{proposition}[Chernoff bound]
    Let $X = \sum_{i=1}^n X_i$ be a random variable where each $X_i$ has a Bernoulli distribution with parameter $p_i$, and let
    $\E[X] = \sum_{i=1}^n p_i$.
    Then for any constant $\delta \ge 0$,
    \[
        \Pr\left[\lvert X-\E[X]\rvert \leq \delta\E[X]\right] \geq 1 - 2e^{-\delta^2\E[X]/3}.
    \]
    If we choose $\delta$ as $10/\sqrt{\E[X]}$ the right term of the inequality becomes a constant very close to $1$, i.e., $> 0.999$.
\end{proposition}

\begin{proposition}[Berry-Essen Theorem for Binomial Distribution]
    Let $X \sim B(n, p)$ be a binomial distributed random variable. Then $Y_n = \frac{X-np}{\sqrt{np(1-p)}}$ converges to $N(0, 1)$,
    i.e., the standard normal distribution.
    Also, for all $x$ and $n$, the difference between $\Pr[Y_n \le x]$ and $\Phi(x)$ is at most $O(1/\sqrt{n})$
    where $\Phi$ denotes the cumulative distribution functions of
    $N(0, 1)$.
\end{proposition}

\section{The Lower Bound}\label{sec:tradeoff}
This section proves the lower bound trade-off between the redundancy and query time for any accessible data structure that
supports both $\adj$ and $\deg$ queries on the interval graph $G$ with $n$ vertices.

\thmtradeoff*

In the above theorem, $r$ represents the redundancy of the data structure since the optimal size for storing all possible inputs is
$\log(n!)$ bits.

In order to prove the lower bound, we begin by defining the tuple of hard queries $\query(i)$ for $i \in [n]$,
which consists of three queries: $\deg(i), \adj(i, e_i), \adj(i, e_i + 1)$.
Note that the answer to $\query(i)$ is $(\deg(i),0,1)$, and recall that $e_i$ is the right endpoint of interval $i$.

To prove \cref{thm:tradeoff}, we will revisit our main idea. Our proof is built on the round elimination strategy outlined
in~\cite{DBLP:conf/soda/PatrascuV10}.
During each round, we divide the vertices of $G$ into $k$ blocks, each consisting of $n/k$ vertices.
Let
$Q_{\Delta} \triangleq \{\query(\Delta+1), \query(n/k + \Delta+1), \query(2n/k+\Delta+1), \dots, \}$  be a set of $k$  queries for
$\Delta \in [0, n/k-2]$, let
$Q'_\Delta\triangleq \{\Delta+1,n/k+\Delta+1,\dots,\}$.
Assuming  the existence of a data structure that can answer $\deg$ and $\adj$ queries in $t_1$ and $t_2$ time, respectively, we can utilize
this structure to reveal $e_i$
and $\deg(i)$ in $t_1+t_2\log n$ time by executing $\deg(i)$ and binary search to locate $e_i$ with $\adj$ queries.
Thus, for any $i\in[n]$, we define $\Bin(i)$ which is the set of cells that are accessed by the aforesaid procedure to reveal $e_i,\deg(i)$.
We define $\Bin(Q),\Bin_\ell(Q)$, $\Bin_{\le\ell}(Q)$, $\Bin_{<\ell}(Q)$,
$\Foot(\Bin(Q))$, $\Foot_\ell(\Bin(Q))$, $\Foot_{<\ell}(\Bin(Q))$ in  the same way with the defining of
$\Probe{}(Q)$, $\Probe{}_\ell(Q)$, $\Probe{}_{\le\ell}(Q)$, $\Probe{}_{<\ell}(Q)$, $\Foot(Q)$, $\Foot_\ell(Q)$, $\Foot_{<\ell}(Q)$.

At the beginning, suppose we have a data structure for answering $\deg$ and $\adj$ using the optimal $\log(n!)$ bits with $P = o(\log(n!))$
published bits. Then we publish all the cells in $\Bin(Q'_0)$.
After publishing the cells in
$\Bin(Q'_0)$, we use the following lemma which shows that publishing the cells in $\Bin(Q'_0)$ decreases the average cell-probe complexity
for answering $\query(i)$ for all
$i \in [n]$ by some constant $c$.
This implies that the total average cell-probe complexity of at least one of $\deg(i),\adj(i,e_i),\adj(i,e_i+1)$ for all vertices $i$ in $G$
is also reduced by $c/3$. The proof of the lemma can be found in \cref{sec:large-overlap}.

\begin{restatable}[Large Overlap]{lemma}{lemqzerogoodoverlap}\label{lem:q0goodoverlap}
    Let $G$ be an interval graph chosen uniformly at random from all possible universal interval representations.
    Let $k \ge \gamma \log^2 n$ for some
    constant $\gamma> 0$, which is at most $n/\polylog(n)$.
    Then for any $(\log(n!) +P)$-bit (accessible) data structure
    on $G$ with $P = o(\log n !)$ published bits,
    \[
        \Pr_{j\in [n]}[\Bin(Q'_0)\cap\Probe{}(\query(j))\ne\emptyset] = \Omega(1).
    \]
\end{restatable}

We adopt an iterative approach whereby we publish the cells in $\Bin(Q'_0)$ one after the other. It is important to note that before
publishing a total of $\log(n!)$ bits, the
average cell-probe complexity cannot possibly be below zero. Based on this observation, we are able to derive the desired lower bound. We
give a complete proof of
\cref{thm:tradeoff} assuming that \cref{lem:q0goodoverlap} is true.

\begin{proof}[Proof of \cref{thm:tradeoff}]
    The proof of our theorem relies on a similar argument to that presented in~\cite{DBLP:conf/soda/PatrascuV10}. Specifically, we iterate
    over a sequence of rounds, with each
    round publishing a certain number of bits. We denote by $P_i$ the total number of bits published up to and including the $i$-th round.

    Initially, we have $P_0 = r$ published bits, and for each round $i \in \{0, 1, \dots \}$, we set $k_i = \gamma P_i \cdot \log^2 n$ for
    some constant $\gamma > 0$, and
    publish the cells in $\Bin(Q'_0)$ with respect to the blocks of size $n/k_i$. We then have
    $P_{i+1} \leq P_i + k_i w (t_1 + t_2\log n)  = O(P_i\log^5 n)$ bits.
    Here, we assume that $t_1 + t_2\log n$ is at most $\log^2n / \log \log n$ since the theorem statement holds trivially if $t_1+t_2$ is
    greater than $\log n / \log \log n$.

    According to \cref{lem:q0goodoverlap}, after each round of bit publishing,
    the total cell-probe complexity required to answer one $\deg$ and two $\adj$ queries decreases by $\Omega(1)$.
    Since the total
    redundancy of the
    published bits cannot exceed $\log(n!)$ bits and average cell-probe complexity cannot be less than zero, we have
    $P_{t_1+t_2} = r (\log n)^{O(t_1+t_2)} \leq \log(n!) = n \log n - \Theta(n)$. This inequality completes the proof of the theorem.
\end{proof}

\subsection{Proving the Large Overlap}\label{sec:large-overlap}
In this section, we prove \cref{lem:q0goodoverlap}. The proof is based on the encoding argument used
in~\cite{DBLP:conf/soda/PatrascuV10,DBLP:conf/stoc/LiuY20}.
Before giving a complete proof, we first describe the overall idea.

To begin the proof by contradiction, assume that $\Pr_{j\in [n]}[\Bin(Q'_0)\cap\Probe{}(\query(j))\ne\emptyset] \le \epsilon^4$ for a
sufficiently small
$0 < \epsilon < 1$.
Using an average argument, we can find a $\Delta$ such that $\Pr_{q\in Q_{\Delta}}[\Bin(Q'_0)\cap\Probe{}(q)\ne\emptyset] \le \epsilon$.
Now let $\ans(\query(i))$ denote the number of right endpoints before the position $i$.
We will use the notation of $q_i$ and $q'_i$ consistently to denote queries in $Q_0$ and $Q_\Delta$ respectively.
We also define $\ans(Q_0)$ and $\ans(Q_{\Delta})$ as $\{\ans(q_i) \mid q_i \in Q_0\}$ and $\{\ans(q'_i) \mid q'_i \in Q_{\Delta}\}$,
respectively.

Since the input graph $G$ is chosen uniformly at random, both $\ans(q)$ and $\ans(Q_p)$ are random variables. To handle them easily, we
introduce an indicator random variable $T(i)$ which we define as
$T(i)\triangleq \mathbf{1}_{\ans(\query(i+n/k ))-\ans(\query(i))\ge t_i}$ for some threshold $t_i > 0$, which we decide later. Also, we
define a vector random variable $T(Q_p)$ as $(T(p+1), T(n/k + p+1) , \dots)$.
Notice that $T(Q_p)$ can be computed from $\ans(Q_p)$.
We introduce the following lemma that shows that the two vector random variables $T(Q_0)$ and $T(Q_{\Delta})$ are highly correlated.

\begin{restatable}[Mutual Information]{lemma}{lemmutual}\label{lem:mutual}
    For any $i\in [0.1k,0.9k]$ and $k\in \omega(1)$,
    suppose $Q_{0} = \{q_1, q_2, \dots, q_k\}$ and $Q_{\Delta} = \{q'_1, q'_2, \dots, q'_k\}$, where $q_i$ and $q'_i$ are $\query(in/k+1)$ and
    $\query(in/k + \Delta+1)$, respectively. Then, for any $\Delta \in \{(1-\epsilon^3)n/k, \dots, n/k-2\}$, the following holds:
    \begin{enumerate}[label=(\roman*)]
        \item $H(T(q_i) \mid T(q_1),T(q_2), \dots, T(q_{i-1})) = \Omega(1)$, but\label{lem:item:high_infor}
        \item $H(T(q_{i+1}) \mid T(q'_i))  = O(\epsilon)$.\label{lem:item:cond_infor}
    \end{enumerate}
\end{restatable}

The proof of \cref{lem:mutual} can be found in \cref{sec:mutual} and is split into \cref{sec:mutual_first,sec:mutual_second}.

For any $\ell \in [t_1+t_2\log n]$, let $\Foot_{<\ell}(\Bin(Q'_0))$ denote the contents of the first $\ell-1$ probed cells. Let
$\Foot_{\ell}(\Bin(Q'_0))$ denote the contents
of just the $\ell$-th probed cell in $\Bin(Q'_0)$.
We begin by encoding the initial $P$ published bits, $\Delta$, and $\ell$ which is uniformly chosen at random.
Next, we write down all the contents of the cells in the database other than $\Foot_{\ell}(\Bin(Q'_0))$.
Notice that we have written down $\Foot_{<\ell}(\Bin(Q'_0))$, which can be used to reveal $\Bin_\ell(Q_0)$
but not the contents of $\Foot_{\ell}(\Bin(Q'_0))$ itself.
Then by the assumption and the chain rule of mutual information, we can reveal $(1-\epsilon/(t_1+t_2\log n))$-fraction of $T(Q_{\Delta})$ on
average  from  the current bits that we have written down.

Based on \cref{lem:mutual}, we observe that this $(1-\epsilon/(t_1+t_2\log n))$-fraction of $T(Q_{\Delta})$ contains sufficient information
about $T(Q_0)$ on average.
This also provides a lower bound on the mutual information between $\Foot_{\ell}(\Bin(Q'_0))$ and $T(Q_{\Delta})$.
Therefore, it suffices to write down the unknown parts of $T(Q_{\Delta})$, then use an optimal encoding scheme to write down
$\Foot_{\ell}(\Bin(Q'_0))$ while conditioning
on the current bits that we have written down and $T(Q_{\Delta})$. This can save more than $P$ bits of space.
In conclusion, we obtain an encoding of the database using strictly less than $\log(n!)$ bits overall, which contradicts to fact that the
information-theoretic lower bound of the database is $\log(n!)$ bits.

Let us now restate \cref{lem:q0goodoverlap} and give a complete proof assuming that \cref{lem:mutual} is true.

\lemqzerogoodoverlap*
\begin{proof}[Proof of \cref{lem:q0goodoverlap}]
    To begin a proof by contradiction, assume that
    \[
        \Pr_{j\in [n]}[\Bin(Q'_0)\cap\Probe{}(\query(j))\ne\emptyset] \le \epsilon^4
    \]
    for a sufficiently small $0 < \epsilon < 1$.
    Then, there exists $\Delta\in \{\floor{(1-\epsilon^3)n/k},\dots, n/k-2\}$ which satisfies
    $\Pr_{q\in Q_{\Delta}}[\Bin(Q'_0)\cap\Probe{}(q)\ne\emptyset] \le \epsilon$, by an averaging argument.
    Now we give an encoding of the data structure that supports $\deg$ and $\adj$ on $G$ as
    follows:

    \begin{enumerate}
        \item Write down the published bits of size $P$ and write down $\Delta$\label{item:encoding:published}.
        \item Sample $\ell\in[t_1+t_2\log n]$ uniformly at random and write it down.\label{item:encoding:2}
        \item Write down $\Foot{}_{<\ell}(\Bin{}(Q'_0))$ with $w\cdot|\Bin_{<\ell}(Q'_0)|$ bits.
        \item Write down all the contents of remaining cells according to the increasing order of their addresses,
              except $\Foot{}_{\ell}(\Bin{}(Q'_0))$.\label{item:encoding:writeparts}
              Note that by enumerating queries $q$ and simulating query algorithm on $q$, one can recover some entries of $\ans(Q_\Delta)$.
        \item Write down the unknown entries of $T(Q_\Delta)$ given the known entries of $\ans(Q_\Delta)$.\label{item:encoding:leftQ0}
        \item Write down $\Foot{}_{\ell}(\Bin{}(Q'_0))$ with any optimal encoding,
              conditioning on $T(Q_\Delta)$.\label{item:encoding:end}
    \end{enumerate}

    The above encoding allows us to reconstruct the data structure as follows.
    Firstly, we decode the values of $P$, $\Delta$, and $\ell$.
    Next, recall that $\Bin(Q'_0)$ is the set of cells probed by $\deg(i)$ and by binary searching $e_i$ for $i\in Q'_0$ with $\adj$ queries,
    so for any
    $\ell\in[t_1+t_2\log n]$, given $\Foot_{<\ell}(\Bin(Q'_0))$ we can compute $\Bin_\ell(Q'_0)$ by simulating the procedure up to step
    $\ell-1$ using $\Foot{}_{<\ell}(\Bin{}(Q'_0))$.
    We then decode the contents of all remaining cells not in $\Bin_\ell(Q'_0)\setminus \Bin_{<\ell}(Q'_0)$.

    Next, we enumerate all queries $\deg(\cdot),\adj(\cdot,\cdot)$ and simulate the query algorithm on the queries with the decoded
    information.
    When we enumerate and successfully simulate the query algorithm on a hard query tuple $\deg(i),\adj(i,e_i),\adj(i,e_i+1)$, we recover an
    answer $\ans(q_i)$.
    Putting all the recovered $\ans(q_i)$ together, we know some entries of $T(Q_\Delta)$.
    Therefore, we can put the known entries together with the missing entries which are written down in Step 5, then we can fully recover
    $T(Q_\Delta)$.
    After decoding $T(Q_{\Delta})$ completely, we can decode $\Foot{}_{\ell}(\Bin{}(Q'_0))$ by referring to the encoding in Step 6.
    At this point, we have decoded every single cells of the data structure.
    Therefore, the total number of bits we have written down in the encoding procedure must be at least the information-theoretic minimum of
    $G$, which is $\log(n!)$ bits.

    We analyze the space of the encoding. Here, we use $\mathcal{P}$ to denote the random variable of the published bits, which implies
    $|\mathcal{P}| = P$.
    \begin{align*}
         & P +  O(\log {\frac{\epsilon n}{k}} )                                                              & \text{(Step 1)} \\
         & + O(\log \log n)                                                                                  & \text{(Step 2)} \\
         & + w\cdot\lvert \Bin{}_{<\ell}(Q'_0)\rvert                                                         & \text{(Step 3)} \\
         & + \log (n!) - w\cdot\lvert \Bin{}_{\le\ell}(Q'_0)\rvert                                           & \text{(Step 4)} \\
         & + O(\epsilon k/(t_1+t_2\log n))                                                                   & \text{(Step 5)} \\
         & + w\cdot|\Bin_\ell(Q'_0)\setminus \Bin_{<\ell}(Q'_0)|                                             & \text{(Step 6)} \\
         & - I(\Foot{}_{\ell}(\Bin{}(Q'_0)); T(Q_\Delta) \mid \Foot{}_{<\ell}(\Bin(Q'_0)),\ell, \mathcal{P}) & \text{(Step 6)}
    \end{align*}
    The following two claims are used to analyze the space for Step 5 and 6.
    We first analyze the space at Step~\ref{item:encoding:end}.
    \begin{claim}\label{cl:encoding:mutual_information}
        \[
            \E[I(\Foot{}_{\ell}(\Bin{}(Q'_0)); T(Q_\Delta) \mid \Foot{}_{<\ell}(\Bin(Q'_0)),\ell, \mathcal{P})]
            \geq \Omega\left(\frac{k}{t_1+t_2\log n}\right).
        \]
    \end{claim}
    \begin{proof}[Proof of \cref{cl:encoding:mutual_information}]
        Since $T(Q_0)$ can be computed from the cells in $\Bin(Q'_0)$, we obtain the following.
        \pagebreak
        \begin{align*}
             & I\left(\Foot{}(\Bin{}(Q'_0));T(Q_\Delta) \right)                                                                                                              \\
             & \geq I(T(Q_0); T(Q_\Delta))                                                                                                                                   \\
             & =  H(T(Q_0)) - H(T(Q_0) \mid T(Q_\Delta))                                                                                                                     \\
             & =\sum_{j \in [k]} H(T(q_j) \mid T(q_1), \dots, T(q_{j-1})) - \sum_{j \in [k]} H(T(q_j) \mid T(q_1), \dots, T(q_{j-1}) ; T(q'_{1}), T(q'_2), \dots T(q_{j-1})) \\
             & \ge \sum_{j \in [k]} H(T(q_j) \mid T(q_1), \dots, T(q_{j-1})) - H(T(q_1)) - \sum_{j \in [0.1k, 0.9k]} H(T(q_{j+1}) \mid T(q'_{j}))                            \\
             & \ge \Omega(k) - \Omega(1) - O(\epsilon k) = \Omega(k).
        \end{align*}
        The last inequality is obtained from \cref{lem:mutual}. Notice that we can consider
        $I(\Foot(\Bin(Q'_0));T(Q_\Delta))$ as $t_1+t_2\log n$ random variables. Then we
        obtain the claim from the chain rule of mutual information, $I(X_1,\dots,X_n ; Y) = \sum_i I(X_i;Y\mid X_{<i})$.
    \end{proof}

    Next, we prove the space used at Step~\ref{item:encoding:leftQ0}.
    \begin{claim}\label{cl:encoding:entropy}
        We write down $O(\epsilon k/(t_1+t_2\log n))$ bits in Step 5.
    \end{claim}
    \begin{proof}[Proof of \cref{cl:encoding:entropy}]
        By an average argument and the assumption the following holds:
        \begin{align}\label{eq:averaging}
            \E{}\left[\left|\lbrace q\in Q_\Delta : \Probe{}(q)\cap \Probe{}_\ell(\Bin{}(Q'_0))\ne\emptyset\rbrace\right|\right]
            \le \epsilon k/(t_1+t_2\log n).
        \end{align}
        This implies that we can recover at least $(1-\epsilon/(t_1+t_2\log n))$-fraction of $\ans(Q_\Delta)$ in expectation with the
        information we have written down in
        Step~\ref{item:encoding:2}-\ref{item:encoding:writeparts}.
        By a simply counting, we can recover at least $(1-2\epsilon/(t_1+t_2\log n))$-fraction of $T(Q_\Delta)$ in expectation.
        The proof is completed by noting that the entries of $T(Q_\Delta)$ are bits and $|T(Q_\Delta)|=k$.
    \end{proof}

    In overall, the total expected space of the encoding is at most:
    \begin{align}\label{eq:overall_size}
         & O\left(\log {\frac{\epsilon n}{k}}\right) + O(\log \log n) + \log (n!) +
        O\left(\frac{\epsilon k}{t_1 + t_2\log n}\right) - \Omega\left(\frac{k}{t_1+t_2\log n}\right).
    \end{align}

    For a sufficiently large constant $\gamma > 0$, if $k \geq \gamma \log^2 n$, then for small enough constant $\epsilon$, the inequality
    given above implies that the
    value is strictly less than $\log(n!)$. This allows us to construct a data structure that uses less space than the information-theoretic
    minimum of $G$, which leads to a contradiction.
\end{proof}

\subsection{Analysis of the Mutual Information}\label{sec:mutual}
In this section, we prove \cref{lem:mutual}.
\lemmutual*
We begin with describing a high-level idea of the proof. Recall that we divide the universe $[n]$ into $k$ blocks, each is of length $n/k$.
Also, for each $q_t \in Q_0$ and $q'_t \in Q_{\Delta}$, we denote $T(q_t)$ and $T(q'_t)$ as $T((t-1)n/k + 1)$ and $T((t-1)n/k + \Delta +1)$,
respectively (recall
that $T(i)$ is an indicator random variable defined as $T(i)\triangleq \mathbf{1}_{\ans(\query(i+n/k))-\ans(\query(i))\ge t_i}$).
Then for each block $i \in \{1, \dots, k\}$, we first define the following random variables:
\begin{itemize}
    \item $Y_i$: the number of right endpoints that land before position $in/k$.
    \item $Y'_i$: the number of intervals whose right endpoints that land in the $i$-th block, and left endpoints land before block $i$.
    \item $Z_i$: the number of intervals whose right endpoints that land in the $i$-th block, and left endpoints land in block $i$.
    \item $A_i$: the number of right endpoints that land in the $i$-th block, i.e., $A_i=Y'_i+Z_i$.
\end{itemize}

Now in the rest of the paper, we fix the threshold $t_i$ as the average number of right endpoints of the intervals
in the $i$-th block, i.e., $t_i = \E[Y'_i + Z_i]$. Then we can observe the followings:
\begin{enumerate}
    \item $T(q_i)$ only depends on $Y'_i$ and $Z_i$.
    \item $Z_i$ is independent of $T(q_1), T(q_2), \dots, T(q_{i-1})$.
    \item $T(q_1),T(q_2), \dots, T(q_i)$ can be computed from $Y_1,\dots, Y_{i+1}$.
    \item $Y'_i$ is independent of $Y_1,\dots, Y_{i-1}$, given $Y_i$.
\end{enumerate}
Hence,
\[
    H(T(q_i) \mid T(q_1),T(q_2), \dots, T(q_{i-1}))\ge H(\mathbf{1}_{Y'_i+Z_i\ge t_i} \mid Y_1,\dots, Y_{i})= H(\mathbf{1}_{Y'_i+Z_i\ge t_i}\mid Y_i).
\]

To prove \cref{lem:mutual}~\ref{lem:item:high_infor}, it suffices to show that $H(T(q_i) \mid Y_i) = \Omega(1)$ from the above inequality. We
first show that $Y_i$ is
$\Theta(in/k)$ with high probability using a Chernoff bound. Then when $Y_i$ is fixed to $m = c \cdot in/k$ for some constant $c >0$, we show
that both $\Pr[A_i \ge t_i \mid Y_i=m]$ and  $\Pr[A_i \le t_i \mid Y_i=m]$ are $\Omega(1)$, which proves the first part of the lemma (recall
that $A_i$  is $Y'_i + Z_i$, which
decides $T(q_i)$). Note that  $Y'_i$ has a binomial distribution, given $Y_i$, whereas $Z_i$ does not, which means we cannot directly use the
idea of the proof in~\cite{DBLP:conf/soda/PatrascuV10}. To resolve the issue, we first show that with high probability, the information
whether $Y'_i$ is at least its expectation or not, dominates the value of $T(q_i)$, and show that the information-theoretic lower bound of
such information is $\Omega(1)$.

For the proof of \cref{lem:mutual}~\ref{lem:item:cond_infor}. We show that with high probability, both
values the $T(q_{i+1})$ and $T(q'_i)$ are basically determined by the number of right endpoints whose left endpoints land before the
$(i+1)$-th block, and the number of right endpoints land in $\{(i+1)n/k + 1, \dots, (i+1)n/k + \Delta + 1\}$
(the `common' part of $T(q_{i+1})$ and $T(q'_i)$).
In the rest of this section, we assume that $i\in [0.1k,0.9k]$.

\subsubsection{Proof of \cref{lem:mutual}~\ref{lem:item:high_infor}}\label{sec:mutual_first}
We begin with the following lemma, which shows that both $\E[Y_i]$ and $in/k-\E[Y_i]$ are $\Omega(in/k)$.
\begin{lemma}\label{lem:mui_bound}
    $\E[Y_i]$ is $O(in/k)$ and is at least $i^2n/(4k^2)$.
\end{lemma}
\begin{proof}
    Since $Y_i$ cannot be greater than $in/k-1$, it is clear that $\E[Y_i]$ is at most $O(in/k)$. Thus, it suffices to show that the lower
    bound of $\E[Y_i]$ is also in $\Theta(in/k)$. To prove this, we use an averaging argument as follows:
    \begin{align*}
        \E[Y_i] & \ge \sum_{j=1}^{in/(2k)}\Pr[\text{right endpoint of interval }j\text{ lands in }[in/(2k),in/k]] \\
                & \ge \frac{in}{2k}\cdot \frac{in/(2k)}{n}=\frac{i^2n}{4k^2}\label{EY_upper}.
    \end{align*}
    Since we only consider the case that $i$ is $\Theta(k)$, $\frac{i^2n}{4k^2}$ is in $\Theta(in/k)$.
\end{proof}
According to the Chernoff bound, the random variable $Y_i$ lies within the range
$\Big[\E[Y_i]-10\sqrt{\E[Y_i]},\E[Y_i]+10\sqrt{\E[Y_i]}\Big]$ with a
probability greater than
$0.999$. Therefore, it suffices to show that $H(T(q_i) \mid Y_i = m)$ is $\Omega(1)$ when
$m \in \left[\E[Y_i]-10\sqrt{\E[Y_i]}, \E[Y_i]+10\sqrt{\E[Y_i]}\right]$.
As previously explained, we can establish the lemma by proving
$\Pr[A_i \ge t_i \mid Y_i=m]$ and  $\Pr[A_i \le t_i \mid Y_i=m]$ are $\Omega(1)$. Since $\E[A_i] = \E[Y'_i] + \E[Z_i]$ and $Y'_i$ and $Z_i$
are independent,
$\Pr[A_i\ge t_i \mid Y_i=m]$ is at least
\[
    \Pr\biggl[Y'_i \ge \E[Y'_i] \land Z_i\ge\E[Z_i] \mid Y_i = m\biggr] \ge \Pr\biggl[Y'_i \ge \E[Y'_i] \mid Y_i = m\biggr] \cdot \Pr[Z \ge \E[Z_i]].
\]
Thus, if (i) $\Pr[Y'_i \ge \E[Y'_i] \mid Y_i = m]$ is $\Omega(1)$, and (ii) $\Pr[Z \ge \E[Z_i]]$ is $\Omega(1)$, we can prove
$\Pr[A_i\ge t_i \mid Y_i=m]$ is $\Omega(1)$.

Unfortunately, for large values of $k$, the property (ii) does not hold, even if $k$ is asymptotically smaller than $n$. For instance, when
$k$ is
$\Theta(n/\log n)$, the reciprocal of the standard deviation of $Z_i$ is $\omega(1)$, which causes $\Pr[Z \geq \E[Z_i]]$ to be $o(1)$. Note
that same issue occurs when we
compute $\Pr[A_i \le t_i \mid Y_i=m]$ by computing $\Pr[Y'_i < \E[Y'_i] \mid Y_i = m]$ and $\Pr[Z < \E[Z_i]]$ independently.

To overcome this problem, rather than computing the two probabilities independently, we establish the subsequent lemma motivated
from~\cite{DBLP:conf/soda/PatrascuV10}.
This lemma proves that the probability of $Y'_i$ deviating significantly larger from its expected value is much
higher than that of $Z_i$ undergoing the same occurrence.

\begin{lemma}\label{lem:Y_prime_and_Z}
    For any $k = \omega(1)$, the following holds for any small constant $0 < \kappa < 1$ that satisfies $k = \omega(\kappa^{-3})$:
    \begin{enumerate}[label=(\roman*)]
        \item $\Pr\left[Y'_i-\E[Y'_i]\ge \kappa\cdot\sqrt{\E[Y'_i]} \mid Y_i = m\right] = 1/2-O(\kappa)$\label{Y_first},
        \item $\Pr\left[Y'_i-\E[Y'_i]\le -\kappa\cdot\sqrt{\E[Y'_i]} \mid Y_i = m\right] = 1/2-O(\kappa)$, and\label{Y_second}
        \item $\Pr\left[\lvert Z_i-\E[Z_i]\rvert\ge \kappa\cdot\sqrt{\E[Y'_i]}\right] =  O(\kappa)$\label{Y_third}.
    \end{enumerate}
\end{lemma}
Notice that for some appropriate choice of $\kappa$, we can ensure that $1/2-O(\kappa)$ is close to $1/2$ and that $O(\kappa)$ is small.
\begin{proof}
    Let $n'=in/k - \E[Y_i] = in/k-m$. Then by \cref{lem:mui_bound}, $c \cdot in/k < n' < d \cdot in/k$ for some constant $0<c<d<1$.
    By \cref{prop:independent}, $Y'_i$ has a binomial distribution
    $B\left(n',\frac{n/k}{n-in/k}\right) = B\left(n',\frac{1}{k-i}\right)$ when $Y_i =m$.
    \begin{enumerate}[label=(\roman*)]
        \item We use $p$ to denote $1/(k-i)$. Then by the Berry-Essen theorem and the Taylor expansion,
              $\Pr\left[Y'_i-\E[Y'_i]\ge \kappa\cdot\sqrt{\E[Y'_i]} \mid Y_i = m\right]$
              is approximated by
              $1-\left(\Phi\left(\frac{\kappa}{2 \pi \sqrt{1-p}}\right)+O\left(\sqrt{1/n'}\right)\right) \approx 1/2 - O\left(\frac{\kappa}{\sqrt{1-p}}\right)$.
              Since $p$ is $o(1)$ for any
              $k = \omega(1)$, $O\left(\frac{\kappa}{\sqrt{1-p}}\right)$ is
              $O(\kappa)$, which completes the proof.

        \item This can be proved by the same argument as the proof of \cref{lem:Y_prime_and_Z}~\ref{Y_first}.

        \item For each interval $j + in /k$ in the $i$-th block,
              its right endpoint also exists in $i$-th block with probability
              $p_j = \frac{n/k -j}{n - in/k - j}$. Then by \cref{prop:independent}, $Z_i$ has a Poisson binomial distribution with a parameter
              $(p_0, \dots, p_{n/k-1})$. Not let $\sigma_{Z_i}$ be a standard deviation of $Z_i$. Since
              $\sigma_{Z_i}^2 = \sum_{j=0}^{n/k-1} (1-p_j)p_j \le \sum_{j=0}^{n/k-1} p_j $, and
              $p_j \le 1/(k-i-1)$ for all $j$, $\sigma_{Z_i}^2$ is at most $\frac{n}{k(k-i-1)}$. By the Chebyshev's inequality, we obtain the following:
              \begin{align*}
                   & \Pr\left[\lvert Z_i-\E[Z_i]\rvert\ge \kappa\cdot\sqrt{\E[Y'_i]}\right] \le \frac{\sigma_{Z_i}^2}{\kappa^2 \cdot \E[Y'_i]} \\
                   & \le \frac{\frac{n}{k(k-i-1)}}{\kappa^2 \cdot \frac{1}{k-i} \cdot c\frac{in}{k}}
                  \leq \frac{k-i}{k-i-1} \cdot \frac{1}{\kappa^2 \cdot c \cdot i }
                  \le \frac{2}{\kappa^2  \cdot i \cdot c}.\label{eq:cheby_on_Z}
              \end{align*}
              The above implies that for any $k = \omega(\kappa^{-3})$, $\Pr[\lvert Z_i-\E[Z_i]\rvert\ge \kappa\cdot\sqrt{\E[Y'_i]}]$ is $O(\kappa)$
              since $i = \Theta(k)$. \qedhere
    \end{enumerate}
\end{proof}
From \cref{lem:Y_prime_and_Z},
we can prove \cref{lem:mutual}~\ref{lem:item:high_infor} as follows.
Let $T_Y(q_i)$ be an indicator random variable defined as $T_Y(q_i) = 1$ if and only if $Y'_i \ge \E[Y'_i \mid Y_i = m]$.
Then by \cref{lem:Y_prime_and_Z}~\ref{Y_first} and~\ref{Y_second}, it is clear that $H(T_Y(q_i))$ is $\Omega(1)$. On the other hand, by
\cref{lem:Y_prime_and_Z} and the reverse union bound, $Y'_i$ deviates from expectation more than $\kappa\sqrt{\E[Y'_i]}$ whereas $Z_i$
deviates from expectation at most $\kappa\sqrt{\E[Y'_i]}$
with probability $1-O(\kappa)$. Hence, with probability $1-O(\kappa)$, $T(q_i) = T_Y(q_i)$, which proves
\cref{lem:mutual}~\ref{lem:item:high_infor}.

\subsubsection{Proof of \cref{lem:mutual}~\ref{lem:item:cond_infor}}\label{sec:mutual_second}

The proof of \cref{lem:mutual}~\ref{lem:item:cond_infor} follows a similar approach as in the proof of Lemma~\ref{lem:Y_prime_and_Z}
with some modifications. Let $i+1$ denote $i'$, and $S_{Y'_{i'}}$ be a set of intervals counted by $Y'_{i'}$. We then partition $S_{Y'_{i'}}$
into two sets $S_1$ and $S_2$ where
\begin{enumerate*}[label=(\roman*)]
    \item $S_1$ is a set of intervals whose right endpoints exist in the range $\{(i+1)n/k + 1, \dots, (i+1)n/k + \Delta + 1\}$, and
    \item $S_2 = S_{Y'_{i'}} \setminus S_1$.
\end{enumerate*}
We refer to $S_1$ as the \emph{common part} of $Y'_{i'}$, since all intervals in this set are used to compute both $T(q_{i+1})$ and
$T(q'_{i})$.
Conversely, we refer to $S_2$ as the \emph{distinct part} of $Y'_{i'}$, since all intervals in this set are only used to compute $T(q_{i+1})$
and not $T(q'_{i})$. We then define two random variables $\tilde Y'_{i'}$ and $Y''_{i'}$, which represent the number of intervals such that
their left endpoints land before block $i$, and their right endpoints are in the common and distinct parts of $Y'_{i'}$, respectively.

Similar as in the proof of \cref{lem:mutual}~\ref{lem:item:high_infor}, we consider $Y_{i'}$ to be fixed to $m$ in
$\Big[\E[Y_{i'}]-10\sqrt{\E[Y_{i'}]},\E[Y_{i'}]+10\sqrt{\E[Y_{i'}]}\Big]$, and denote $\frac{i'n}{k} -m$ as $n'$.
Then $c \cdot \frac{i'n}{k} < n' < d \cdot \frac{i'n}{k}$ for some constants $0<c<d<1$ by \cref{lem:mui_bound}.
We introduce the following lemma, which shows that the probability of $\tilde Y'_{i'}$ is away from its expectation is much larger
than the probability that $Y''_{i'}$ or $Z_{i'}$ are away from their expectations, even in the case that the size of the common
part of $Y'_{i'}$ is minimal (that is, $\Delta = (1-\epsilon^3) \frac{n}{k}$).

\begin{lemma}\label{lem:common_and_distinct}
    Suppose $k = \omega(1)$ and $\Delta = (1-\epsilon^3) \frac{n}{k}$.
    Then for any constant $\kappa$ which satisfies $\sqrt[3]{\epsilon^3/(1-\epsilon^3)} = O(\kappa)$ and $k = \omega(\kappa^{-3})$, the
    following holds:
    \begin{enumerate}[label=(\roman*)]
        \item $\Pr\left[\lvert \tilde{Y}_{i'}'-\E[\tilde{Y}_{i'}']\rvert\ge \kappa\cdot\sqrt{\E[\tilde{Y}_{i'}']} \mid Y_{i'}=m\right] =  1-O(\kappa)$,\label{item:common_and_distinct_yi}
        \item $\Pr\left[\lvert Y''_{i'}-\E[Y''_{i'}]\rvert\ge \kappa/2\cdot\sqrt{\E[\tilde{Y}_{i'}']} \mid Y_{i'} =  m\right] =  O(\kappa)$, and\label{item:common_and_distinct_ypp}
        \item $\Pr\left[\lvert Z_{i'}-\E[Z_{i'}]\rvert\ge \kappa/2\cdot\sqrt{\E[\tilde{Y}_{i'}']}\right] =  O(\kappa)$.\label{item:common_and_distinct_zi}
    \end{enumerate}
\end{lemma}
Notice that our restriction on $\epsilon$
gives us an upper bound on $\epsilon$ depending on $\kappa$.
\begin{proof}
    \begin{enumerate}[label=(\roman*), leftmargin=*]
        \item $\tilde{Y}_{i'}' \sim B(n', \tilde p)$ where $\tilde p =\frac{(1-\epsilon^3)n/k}{n-i'n/k} = \frac{1-\epsilon^3}{k-i'}$ by
              \cref{prop:independent}.
              Then by the same argument as the proof of \cref{lem:Y_prime_and_Z}~\ref{Y_first},
              $\Pr\left[\tilde{Y}_{i'}' \le \E[\tilde{Y}_{i'}'] - \kappa\sqrt{\E[\tilde{Y}_{i'}']}\right]$ is $1/2 - O(\kappa)$
              when $k$ is $\omega(1)$.
        \item From \cref{prop:independent},
              $Y''_{i'} \sim B\left(n', \frac{\epsilon^3}{1-\epsilon^3}\tilde p\right)$, which implies the variance of $Y''_{i'}$ is
              $\frac{\epsilon^3}{1-\epsilon^3}n' \tilde p\left(1-\frac{\epsilon^3}{1-\epsilon^3}\tilde p\right)$. Then by the Chebyshev's
              inequality, we obtain the following when $\sqrt[3]{\epsilon^3/(1-\epsilon^3)}$ is  $O(\kappa)$ (recall that
              $\E[\tilde{Y}_{i'}']$ is $n' \tilde p$):
              \begin{align*}
                   & \Pr\left[\left|Y''_{i'}-\E[Y''_{i'}]\right|\ge \kappa/2\cdot\sqrt{\E[\tilde{Y}_{i'}']}\right]                      \\
                   & \le \frac{4\frac{\epsilon^3}{1-\epsilon^3}\left(1-\frac{\epsilon^3}{1-\epsilon^3} \cdot \tilde p\right)}{\kappa^2}
                  = \frac{4\frac{\epsilon^3}{1-\epsilon^3}\left(1-\frac{\epsilon^3}{k-i}\right)}{\kappa^2}                              \\
                   & \le \frac{4\epsilon^3}{(1-\epsilon^3)\kappa^2} = O(\kappa).
              \end{align*}
        \item Recall that $Z_{i'}$ has a Poisson binomial distribution with $(p_0, \dots, p_{n/k-1})$ where $p_j = \frac{n/k-j}{n-i'n/k-j}$.
              Now let $\sigma_{Z_{i'}}$ be a standard deviation of $Z_{i'}$. Then $\sigma_{Z_{i'}}^2$ is at most $\frac{n}{k(k-i'-1)}$ (see
              the proof of \cref{lem:Y_prime_and_Z}~\ref{Y_third}).
              Then by the Chebyshev's inequality, we obtain the following:
              \begin{align*}
                   & \Pr\left[\lvert Z_{i'}-\E[Z_{i'}]\rvert\ge \kappa/2\cdot\sqrt{\E[\tilde{Y}_{i'}']}\right]
                  \le \frac{4\sigma_{Z_{i'}}^2}{\kappa^2 \cdot \E[\tilde{Y}_{i'}']}                                   \\
                   & \le \frac{4\frac{n}{k(k-i'-1)}}{(1-\epsilon^3)\kappa^2 \cdot \frac{1}{k-i'} \cdot c\frac{in}{k}}
                  \le \frac{8}{(1-\epsilon^3)\kappa^2  \cdot i' \cdot c}.
              \end{align*}
              The above implies that for any $k = \omega(\epsilon^{-3})$,
              $\Pr\left[\left\lvert Z_{i'}-\E[Z_{i'}]\right\rvert\ge \kappa\cdot\sqrt{\E[\tilde{Y}_{i'}']}\right]$ is $O(\kappa)$.
    \end{enumerate}
\end{proof}

Let us proceed to the final step of the proof. We can observe that the number of right endpoints falling within the $i$-th block is given by
the sum of three random variables: $\tilde{Y}_{i'}'$, $Y_{i'}''$, and $Z_{i'}$.
Similarly, for the interval $[(i-1)n/k+\Delta+1, in/k+\Delta]$, which we refer to as the $(i, \Delta)$-th block, the number of right
endpoints falling within this interval is given by the sum of three random variables: $\tilde{Y_{i'}'}$, $Y_{i,\Delta}''$, and $Z_{i,\Delta}$.
Here, we define $Y_{i,\Delta}''$ as the number of intervals whose left endpoints land before the $(i,\Delta)$-th block and whose right
endpoints land before the $i$-th
block. On the other hand, $Z_{i,\Delta}$ represents the number of intervals whose both left and right endpoints fall within the
$(i, \Delta)$-th block.

Note that we cannot directly apply \cref{lem:common_and_distinct}~\ref{item:common_and_distinct_ypp}
and~\ref{item:common_and_distinct_zi} to $Y''_{i, \Delta}$ and $Z_{i, \Delta}$, as they have different distributions from $Y''_{i'}$ and
$Z_{i'}$, respectively.
However, the following lemma implies that we can still make the same argument as
\cref{lem:common_and_distinct}~\ref{item:common_and_distinct_ypp}
and~\ref{item:common_and_distinct_zi} for $Y''_{i, \Delta}$ and $Z_{i, \Delta}$, by replacing the variances of $Y''_{i'}$ and $Z_{i'}$ in the
proofs into the variances of
$Y_{i,\Delta}''$ and $Z_{i,\Delta}$, respectively:

\begin{lemma}\label{lem:shift_variance}
    (a) $\sigma_{Y''_{i, \Delta}}^2 = O(\sigma_{Y''_{i'}}^2)$, and (b)
    $\sigma_{Z_{i, \Delta}}^2$ is at most $\frac{n}{k(k-i'-1)}$.
\end{lemma}
\begin{proof}
    \begin{enumerate}[label=(\roman*), leftmargin=*]
        \item Recall that $Y''_{i'}$ has a binomial distribution $B(n', \frac{\epsilon^3}{k-i'})$, and
              the difference between the leftmost position of the $i'$-th block and
              the $(i, \Delta)$-th block is $(n/k-\Delta) \le \epsilon^3 \cdot n/k$.
              This implies $Y''_{i, \Delta}$ has a binomial distribution
              $B(n'_{i, \Delta}, \tilde{p}_{i, \Delta})$ where $n'_{i, \Delta} = O(n')$, and
              $\frac{\epsilon^3}{k-i'} \le \tilde{p}_{i, \Delta} \le \epsilon^3 \frac{n/k}{n-(i+\epsilon^3)n/k} = \frac{\epsilon^3}{k-i-\epsilon^3} \le \frac{\epsilon^3 }{2(k-i)}$,
              which proves the lemma.

        \item $Z_{i, \Delta}$ has a Poisson binomial distribution with $p_{j, \Delta} = \frac{n/k-j}{n-(i+\Delta)n/k-j} \le \frac{n/k-j}{n-i'n/k-j}$ for
              $j \in \{0, \dots, n/k-1\}$, which proves the lemma (see the proof of \cref{lem:Y_prime_and_Z}~\ref{Y_third}).
    \end{enumerate}
\end{proof}

Now we prove \cref{lem:mutual}~\ref{lem:item:cond_infor} as follows: We first choose a constant $\kappa$ that satisfies the conditions of
\cref{lem:common_and_distinct}.
Then by \cref{lem:common_and_distinct,lem:shift_variance},
\begin{enumerate*}[label=(\roman*)]
    \item $\tilde{Y}_{i'}'$ deviates from its expectation more than $\kappa\sqrt{\E[\tilde{Y}_{i'}']}$, and
    \item both $Y''_{i'}+Z_{i'}$ and $Y''_{i, \Delta}+Z_{i, \Delta}$ deviate from their expectations at most $\kappa \sqrt{\E[\tilde{Y}_{i'}']}$
\end{enumerate*}
with probability $1-O(\kappa)$ by applying the reverse union bound.
Hence, $T(q_{i+1})$ and $T(q'_{i})$ are the same with probability $1-O(\kappa) = 1- \sqrt[3]{\epsilon^3/(1-\epsilon^3)}$, which is
$1-O(\epsilon)$ for small enough
constant $\epsilon$. This completes the proof of  \cref{lem:mutual}~\ref{lem:item:cond_infor}.

\section{Data Structures for Individual Queries}\label{sec:upper}
In this section, we propose two accessible data structures on interval graph with $n$ vertices as follows.
First, we show that there exists a data
structure that can answer
$\adj$ queries in $O(1)$ time with $O(\sqrt{n} \log n)$-bit redundancy in the RAM model (\cref{thm:improved_adj}).
Afterward, we show that there exists a data structure that can answer $\deg$ queries in $O(1)$ time with $O(n^{2/3} \log n)$-bit redundancy
in the RAM model (\cref{thm:improved_deg}).
Note that all data structures are \emph{non-systematic (encoding)}~\cite{DBLP:journals/siamcomp/FischerH11}. In a non-systematic data
structure, one cannot access the
input graph after preprocessing so that only the data structure can be used for answering queries\footnote{Note that this does not mean that
    one cannot reconstruct the input with any non-systematic data structure. For example, we can reconstruct the input graph using the data
    structure of \cref{thm:improved_adj} by checking all possible $\adj$ queries.}.
In the rest of the paper, we ignore all floors and ceilings which do not affect the main result.

\thmimprovedadj*
\begin{proof}
    We first divide an interval $[1, n]$ into blocks of size $\sqrt{n}$.
    If the vertex $i$ has its right endpoint in the $k$-th block at position $j$, we can store the length of $i$ using the block number $k$
    along with the offset $j$. Then the length of each interval in $G$ is enough to check adjacency since the left endpoint of $I_i$ is fixed
    to $i$. Notice that each interval within the $k$-th block can only contain the right endpoints in the $j$-th block where $j \ge i$.
    Additionally, we only have to store the offset for the last block. Therefore, by using the encoding of
    \cref{lem:DPT} (\cite{DBLP:conf/stoc/DodisPT10})
    to encode the lengths of intervals in the same block, the total space required to store all interval lengths can be expressed as
    follows:
    \begin{align*}
         & \sum_{k=1}^{\sqrt{n}} \underbrace{\sqrt{n}}_{\text{vertices per block}} \left(\underbrace{\log (\sqrt{n}-k)}_{\text{block index}}
        + \underbrace{\log \sqrt{n}}_{\text{offset}}\right) + O(\sqrt{n}\log n).
    \end{align*}
    Let us now show that $\sum_{k=1}^{\sqrt{n}} \sqrt{n}(\log(\sqrt{n}-k)\sqrt{n})$ is $(\log(n!)+ O(\sqrt{n} \log n))$, which proves the
    theorem.
    \begin{align*}
         & \sqrt{n}\sum_{k=1}^{\sqrt{n}} \log ((\sqrt{n}-k)\sqrt{n})                                                                             \\
         & =\sqrt{n}\sum_{k=\sqrt{n}}^1 \log (k\sqrt{n})                                                                                         \\
         & =\sum_{k=\sqrt{n}}^2 \sum_{i=1}^{\sqrt{n}}\log (k\sqrt{n}-i + i) + \sqrt{n}\log(\sqrt{n})                                             \\
         & = \sum_{k=\sqrt{n}}^2 \sum_{i=1}^{\sqrt{n}} \left(\log (k\sqrt{n}-i)+ \log (1+\frac{i}{k\sqrt{n}-i})\right) + \sqrt{n}\log (\sqrt{n}) \\
         & \leq \log (n!) + \sum_{k=\sqrt{n}}^2 \sum_{i=0}^{\sqrt{n}} \log (1+\frac{i}{k\sqrt{n}-i})  +\sqrt{n}\log n                            \\
         & < \log (n!) + \sum_{k=\sqrt{n}}^2 \sum_{i=2}^{\sqrt{n}} \log (1+\frac{i}{k\sqrt{n}-i})  + 2\sqrt{n}\log n.
    \end{align*}
    The summation with $k=1$ is at most $1/2\sqrt{n}\log{n}$.
    \begin{align*}
         & < \log (n!) + \sum_{k=\sqrt{n}}^2\sum_{i=0}^{\sqrt{n}} \frac{i}{(k\sqrt{n}-i)} + O(\sqrt{n}\log{n})     \\
         & \le \log (n!) + \sqrt{n}\sum_{k=\sqrt{n}}^2 \frac{\sqrt{n}}{(k\sqrt{n}-\sqrt{n})} + O(\sqrt{n}\log{n}).
    \end{align*}
    Note that $i \le \sqrt{n}$. Hence, we can bound the above by
    \begin{align*}
         & \log (n!) + \sqrt{n}\sum_{k=2}^{\infty} \frac{1}{k-1} + O(\sqrt{n}\log{n}) \\
         & \le \log (n!) + O(\sqrt{n}\log n).
    \end{align*}
\end{proof}

Notice the data structure of \cref{thm:improved_adj} can support $\deg(i)$ queries in $O(n)$ time by answering the number of corresponding
intervals of vertices in $\{1, \dots, e_i\}$ that
intersect with $[i, e_i]$. Also, in the cell-probe model, we can support $\adj{}$ queries in $O(1)$ probes while using $\log n! + O(1)$ bits
of space (see \cref{sec:cell_adj}).

\thmimproveddeg*
\begin{proof}
    As mentioned in \cref{sec:uppertech},  We first divide the intervals into blocks of size $n^{1/3}$.
    For each interval $i$ in the $k$-th block, then all the intervals that intersect with $i$ are within exactly one of the four following
    cases (See \cref{fig:upper_bound_intervals}):

    \begin{enumerate}[label=(\roman{*})]
        \item the left endpoint lands before the $k$-th block, and the right endpoint lands after the $k$-th block;\label{item:deg:bl1}
        \item the left endpoint lands in the interval $i$;
        \item the left endpoint lands before the $k$-th block, and the right endpoint lands in the $k$-th block;
        \item the left endpoint lands in the $k$-th block but before the position $i$.
    \end{enumerate}

    For any two intervals which are in the same block, the number of intersecting intervals in the case~\ref{item:deg:bl1} are always the
    same. Hence, we can store the number of intervals in the case~\ref{item:deg:bl1} using $O({n}^{2/3}\log n)$ bits in total.
    Next, assume that there are $b_k$ right endpoints that land in the $k$-th block.
    Then the numbers of intervals in other three cases are at most
    \begin{enumerate*}[label=(\roman{*})]\setcounter{enumi}{1}
        \item $n-i$,
        \item $b_k$,
        \item $i-(k-1)n^{1/3}$,
    \end{enumerate*}
    respectively, which implies the total number of intervals in the cases (ii)-(iv) intersecting with $i$ is at most $n-(k-1)n^{1/3} +b_k$.
    Therefore, for each interval $i$ in the $k$-th block, we can encode $\deg(i)$ using the $O({n}^{2/3}\log n)$-bit encoding for storing the
    number of intersecting intervals with the case~\ref{item:deg:bl1}, along with $n^{2/3}$ arrays of integers which store
    the number of intersecting intervals in the cases (ii)-(iv), where the $k$-th array is of universe $[n-(k-1)n^{1/3} +b_k]^{\sqrt[3] n}$.
    We store the sizes of the alphabets of the $n^{2/3}$ arrays with an extra array of $O(n^{2/3}\log n)$ bits, so that we can use
    \cref{lem:DPT} to encode the $n^{2/3}$ arrays.
    Overall, we obtain a data structure of $\sum_{k=n^{2/3}}^{1}n^{1/3} \log {(n-(k-1)n^{1/3} + b_k)} + O(n^{2/3} \log n)$ bits in total.
    Now let $k' = n^{2/3}-k+1$ which we use to replace $b_k$.
    Then the following shows that
    $\sum_{k=n^{2/3}}^{1}n^{1/3} \log {(n-(k-1)n^{1/3} + b_k)} \le \sum_{k=n^{2/3}}^{1} n^{1/3} \log {(kn^{1/3} + b_{k'})}$ is
    $O(n^{2/3} \log n)$, which proves the theorem:

    \begin{align*}
         & \sum_{k=n^{2/3}}^1 n^{1/3}\log( kn^{1/3} + b_{k'})                                                                                 \\
         & =\sum_{k=n^{2/3}}^{1} \sum_{i=0}^{n^{1/3}-1} \log(kn^{1/3} - i + i + b_{k'})                                                       \\
         & =\sum_{k=n^{2/3}}^1 \sum_{i=0}^{n^{1/3}-1} \log(kn^{1/3} - i) + \log \left(1+\frac{i + b_{k'}}{kn^{1/3} - i}\right)                \\
         & = \log (n!)  + \sum_{k=n^{2/3}}^1 \sum_{i=0}^{n^{1/3}-1} \log \left(1+\frac{i + b_{k'}}{kn^{1/3}-i}\right)                         \\
         & \leq \log (n!)  + \sum_{k=n^{2/3}}^{1} n^{1/3}\log \left(1+\frac{n^{1/3} + b_{k'}}{kn^{1/3} - n^{1/3}+1}\right)                    \\
         & \leq \log (n!) + O(n^{1/3} \log n)+ \sum_{k=n^{2/3}}^{2} n^{1/3}\log \left(1+\frac{n^{1/3} + b_{k'}}{kn^{1/3} - n^{1/3}+1}\right).
    \end{align*}
    We used that $n^{1/3}+b_{n^{2/3}}$ is $O(n)$.
    Let us now bound the remaining sum form and set $a_k=n^{1/3}+b_{k'}$.
    \begin{align*}
         & \sum_{k=n^{2/3}}^{2} n^{1/3}\log \left(1+\frac{n^{1/3} + b_{k'}}{kn^{1/3} - n^{1/3}+1}\right)                                   \\
         & \leq n^{1/3}\sum_{\ell=\frac{2}{3}\log n}^0 \sum_{k=2^\ell}^{2^{\ell-1}} \log \left(1+\frac{a_k}{(k-1)n^{1/3}}\right)           \\
         & \leq n^{1/3}\sum_{\ell=\frac{2}{3}\log n}^0 \sum_{k=2^\ell}^{2^{\ell-1}} \log \left(1+\frac{a_k}{(2^{\ell-1}-1)n^{1/3}}\right).
    \end{align*}
    As for all $\ell \in \{0, \dots, \frac{2}{3}\log n \}$, $\sum_{k=2^\ell}^{2^{\ell -1}} a_k = \sum_{k=2^\ell}^{2^{\ell -1}} (n^{1/3} + b_{k'}) \le 2n$ (notice that
    $\sum_{k'} b_{k'} =n$).
    Also, by the concavity of logarithm and Jensen's inequality, $\sum_{i=1}^n \frac{\log (1+x_i)}{n} \leq
        \log (1+(\sum_i x_i)/n)$. Therefore, the above equation is at most
    \begin{equation*}
        n^{1/3} \sum_{\ell=\frac{2}{3}\log n}^0 2^{\ell-1}\cdot\log \left(1+\frac{2n/2^{\ell-1}}{(2^{\ell-1}-1)n^{1/3}}\right).
    \end{equation*}
    Now for $\ell \in \{0, \dots, \frac{2}{3}\log n \}$, we define $g(\ell) = 2^{\ell-1}\log (1+\frac{n^{2/3}}{2^{2\ell-2}})$.
    If $2^{2 \ell-2} \le n^{2/3}$, i.e., $\ell \le \frac{1}{3}\log n + O(1)$, we have $g(\ell) = O(n^{1/3} \log n)$.
    On the other hand, if $2^{2 \ell-2} > n^{2/3}$, i.e., $\ell \ge \frac{1}{3}\log n + \Omega(1)$, we have $g(\ell) \le O(\frac{n^{2/3}}{2^{\ell}})$, as $\log(1+x)\leq x$.
    So we have $g(\ell)\leq 2^{\ell-1} \frac{n^{2/3}}{2^{2\ell-2}}$ in this case. Hence, $g(\ell)$ is maximized when $\ell$ is $\frac{1}{3} \log n \pm O(1)$. In this case,
    the above equation is at most
    $\Theta(n^{2/3} \log n)$.

\end{proof}
Notice that the data structure of \cref{thm:improved_deg} can also support $\adj(i,j)$ queries in $O(n)$ time by completely reconstructing
the graph. This is done as follows.
To address the query $\adj(i,j)$, we utilize the degrees of vertices from $1$ to $\max{}(i, j)$ and decode the corresponding intervals.
Initially, we compute $\deg(1)$,
which allows us to decode $I_1$ as $[1, \deg(1)]$ (note that all intervals in $I$ intersect with $I_1$ should have their left endpoints
in $[2, e_i]$). For $k \in \{2, \dots, \max \{ i, j\} \}$, we can decode $I_k$ using $\deg(k)$ and the intervals ${I_1, \dots, I_{k-1}}$.
Specifically, $I_k = [k, \deg(k)-k']$, where $k'$
represents the number of intervals in ${I_1, \dots, I_{k-1}}$ that cross the point $k$.
Finally, to determine $\adj(i, j)$, we check whether $I_i$ and $I_j$ intersect or not.

\printbibliography

\appendix
\section{Data structure for \texorpdfstring{$\adj{}$}{adj} query in the cell-probe model}\label{sec:cell_adj}
In this section, we give a data structure for $\adj{}$ queries on interval graph with $n$ vertices in the cell-probe model.
We use the following lemma from~\cite{DBLP:conf/stoc/DodisPT10} which gives us more control over the mapping than their main theorem.
\begin{lemma}[{\cite[Lemma 4]{DBLP:conf/stoc/DodisPT10}}]\label{lem:DPTdetail}
    Let $X,Y\leq 2^w$. Then for some values $M$ and $S$ we can injectively map $(x,y)\in [X]\times [Y]$ to $(m,s)\in [2^M]\times [S]$ with
    the following properties:
    \begin{itemize}
        \item $S$ and $M$ are chosen as functions of $X$ and $Y$ but satisfy $S=O(\sqrt{X})$, $2^M=O(Y\sqrt{X})$.
        \item The map can be evaluated in time $O(1)$.
        \item $x$ can be decoded from $m$ and $s$ in $O(1)$ time.
        \item $y$ can be decoded from $m$ alone in $O(1)$ time.
        \item The redundancy is $(M+\log S)-(\log X + \log Y) = O(1/\sqrt{X})$.
    \end{itemize}
\end{lemma}
The details of this can be found in~\cite{DBLP:conf/stoc/DodisPT10}.

\begin{restatable}[]{theorem}{thmimproveddeg2}\label{thm:improved_deg2}
    There is a cell-probe data structure for $n$-vertex interval graphs in universal interval representation, supporting adjacency query in
    constant time using $\log(n!)+O(1)$ bits of space.
\end{restatable}
\begin{proof}
    Let $x_i$ be the length of the interval whose left endpoint is $i$, which is from $[n-i+1]$. We then encode the length of each interval
    iteratively using \cref{lem:DPTdetail}.
    In other words, we would like to encode an array $(x_1,\dots,x_n)\in[n]\times[n-1]\times\dots\times[1]$.

    We start with grouping the numbers in the following way, so that each group can be viewed as an element from a universe of size
    $\Theta(n^2)$.
    For this, the first group consists of the first two numbers, so the first group is of universe $[n]\times[n-1]$ of size $\approx n^2$.
    We iteratively group the elements to reach universes of size $\approx n^2$, and some groups may consist of more than two elements as the
    size of the universe of the element is decreasing.
    At the end of the grouping procedure, we obtain an array of groups such that each group is from a universe of size $\approx n^2$, except
    the last group. Note that each size of the universe of $i$-th group can be computed without knowing anything about the input, as it is a
    function of $i$ and $n$. Also note that the number of the groups, denoted by $n'$, is a function of $n$.

    We encode each group trivially by the multiplication, i.e., $(x_i,\dots,x_j)$ is written as a number
    $\sum_{k=i}^j (x_{k}\prod_{l=k+1}^j (n-l+1))$,\footnote{We let $\prod_{l=j+1}^j(n-l+1)=1$.} resulting in $y_{i'}$, so that we can fully decode $(x_i,\dots,x_j)$ from $y_{i'}$.
    From now on, given the grouping, we want to encode an array $(y_1,\dots,y_{n'})$ such that each element is from a universe of size $\Theta(n^2)$.
    First trivially apply the lemma (\cref{lem:DPTdetail}) with $(x,y)=(y_0,y_1)$, where $y_0\in[1]$ is a dummy element, obtain $(m_1,s_1)$.
    Then we iteratively apply the lemma for $i\in[2,n']$ with $(x,y)=(s_{i-1},y_i)$ and obtain $(m_i,s_i)$.
    We write down memory block $m_i$ for $i\in [n']$ one by one and, finally, $s_{n'}$.

    Let us now analyze the redundancy for $\sum_{i<n'} M_i$. The redundancy is given by $\sum_i O(1/\sqrt{X_i})$, where $X_i=\Theta(n^2)$
    (except $X_{n'}$) is the size of the universe of aforesaid $y_i$ determined by the lemma.
    We can see this sum as $n'\cdot O(1/\sqrt{n^2})=O(n'/n)= O(1)$ as $n'\le n$.
    We use $\lceil\log S_{n'}\rceil$ bits to write down $s_{n'}$, where $[S_{n'}]$ is the universe of $s_{n'}$, that incurs $1$ bits of
    redundancy.

    Now we want to answer a query. To reconstruct the length of interval $i$, we simulate the compression algorithm to find the sizes of the
    groups and the $M_i$s.
    As the sizes of the groups and the $M_i$'s are independent of the actual bits stored, this does not require any knowledge about the input
    and can be done for free in the cell-probe model.
    We then retrieve  the memory block $m_{i'}$ which maintains the $i$-th interval length since $\sum_{i<i'}M_i$ is known.
    We also reconstruct the spill $s_{i'}$ from  the bits $m_{i'+1}$ ($s_{n'}$, if $i'=n'$) which by \cref{lem:DPTdetail} can be
    reconstructed. From this we have the corresponding tuple $(m_{i'},s_{i'})$.
    From this we can decode the tuple $y_{i''}$ which encodes $(x_j,\dots, x_i,\dots, x_k)$.
    As the tuple is encoded trivially, we can retrieve the  actual  number $x_i$.
    This shows a time complexity of $O(1)$ probes in total, as retrieving $m_{i'},m_{i'+1}$ takes $O(1)$ probes.
    As for correctness, knowing the length of interval $i$ allows us to know $\adj(i,\cdot)$.
\end{proof}

\end{document}